\documentclass{jocg}
\usepackage[english]{babel}
\usepackage{graphicx}
\usepackage{color}
\usepackage{amssymb}
\usepackage{amsmath}
\usepackage{amsfonts}
\usepackage{xspace}
\usepackage{wrapfig}
\usepackage{refcount}
\usepackage{caption}
\usepackage{subcaption}
\usepackage{array}

\graphicspath{{figures/}}

\newcommand{\etal}{\emph{et al.}\xspace}
\newcommand{\tb}{\ensuremath{\overline{\theta}}}
\newcommand{\canon}[2]{\ensuremath{D_{#1}^{#2}}}

\usepackage{amsthm}
\theoremstyle{plain}
\newtheorem{theorem}{Theorem}
\newtheorem{lemma}[theorem]{Lemma}
\newtheorem{corollary}[theorem]{Corollary}

\author{
  Luis~Barba,%
  \footnotemark[2]
  \thanks{\affil{D\'epartement d'Informatique, Universit\'e Libre de Bruxelles},
          \email{lbarbafl@ulb.ac.be}}\,
  Prosenjit~Bose,%
  \thanks{\affil{School of Computer Science, Carleton University},
          \email{jit@scs.carleton.ca, andre@cg.scs.carleton.ca, sander@cg.scs.carleton.ca}. Research supported in part by NSERC and Carleton University's President's 2010 Doctoral Fellowship.}\,
  Mirela~Damian,%
  \thanks{\affil{Department of Computing Sciences, Villanova University},
          \email{mirela.damian@villanova.edu}. Research supported by NSF grant CCF-1218814.}\,
  Rolf~Fagerberg,%
  \thanks{\affil{Department of Computer Science, University of Southern Denmark},
          \email{rolf@imada.sdu.dk}}\,
  Wah~Loon~Keng,%
  \thanks{\affil{Department of Computer Science, Lafayette College},
          \email{kengw@lafayette.edu, gexia@cs.lafayette.edu}}\,
  Joseph~O'Rourke,%
  \thanks{\affil{Department of Computer Science, Smith College},
          \email{orourke@cs.smith.edu}}\,
  Andr\'e~van~Renssen,%
  \footnotemark[2]\,
  Perouz~Taslakian,%
  \thanks{\affil{School of Science and Engineering, American University of Armenia},
          \email{perouz.taslakian@ulb.ac.be}}\,\,\,
  Sander~Verdonschot,%
  \footnotemark[2]\,
  and Ge~Xia%
  \footnotemark[5]
}
\title{\MakeUppercase{New and Improved Spanning Ratios for Yao Graphs}}
\date{\today}

\begin{document}

\maketitle

\begin{abstract}
For a set of points in the plane and a fixed integer $k > 0$, the Yao graph $Y_k$ partitions the space around each point into $k$ equiangular cones of angle $\theta=2\pi/k$, and connects each point to a nearest neighbor in each cone. It is known for all Yao graphs, with the sole exception of $Y_5$, whether or not they are geometric spanners. In this paper we close this gap by showing that for odd $k \geq 5$, the spanning ratio of $Y_k$ is at most $1/(1-2\sin(3\theta/8))$, which gives the first constant upper bound for $Y_5$, and is an improvement over the previous bound of $1/(1-2\sin(\theta/2))$ for odd $k \geq 7$.

We further reduce the upper bound on the spanning ratio for $Y_5$ from $10.9$ to \mbox{$2+\sqrt{3} \approx 3.74$}, which falls slightly below the lower bound of $3.79$ established for the  spanning ratio of $\Theta_5$ ($\Theta$-graphs differ from Yao graphs only in the way they select the closest neighbor in each cone). This is the first such separation between a Yao and $\Theta$-graph with the same number of cones. We also give a lower bound of $2.87$ on the spanning ratio of $Y_5$.

In addition, we revisit the $Y_6$ graph, which plays a particularly important role as the transition between the graphs ($k > 6$) for which simple inductive proofs are known, and the graphs ($k \le 6$) whose best spanning ratios have been established by complex arguments. Here we reduce the known spanning ratio of $Y_6$ from $17.6$ to $5.8$, getting closer to the spanning ratio of 2 established for $\Theta_6$.

Finally, we present the first lower bounds on the spanning ratio of Yao graphs with more than six cones, and a construction that shows that the Yao-Yao graph (a bounded-degree variant of the Yao graph) with five cones is not a spanner.
\end{abstract}

\section{Introduction}
\label{sec:introduction}


The complete Euclidean graph defined on a point set $S$ in the plane is the graph with vertex set $S$ and edges connecting each pair of points in $S$, where each edge $xy$ has as weight the Euclidean distance $|xy|$ between its endpoints $x$ and $y$. Although this graph is useful in many different contexts, its main disadvantage is that it has a quadratic number of edges. As such, much effort has gone into the development of various methods for constructing graphs that {\em approximate} the complete Euclidean graph. What does it mean to approximate this graph? One standard approach is to construct a spanning subgraph with fewer edges (typically linear) with the additional property that every edge $e$ of the complete Euclidean graph is approximated by a path in the subgraph whose weight is not much more than the weight of $e$. This gives rise to the notion of a {\em t-spanner}. A $t$-spanner of the complete Euclidean graph is a spanning subgraph with the property that, for each pair of vertices $x$ and $y$, the weight of a shortest path in the subgraph between $x$ and $y$ is at most $t \geq 1$ times $|xy|$. The \emph{spanning ratio} is the smallest $t$ for which the subgraph is a $t$-spanner. Spanners find many applications, such as approximating shortest paths or minimum spanning trees. For a comprehensive overview of geometric spanners and their applications, we refer the reader to the book by Narasimhan and Smid \cite{NS06}.

\begin{figure}[ht]
 \centering
 \includegraphics{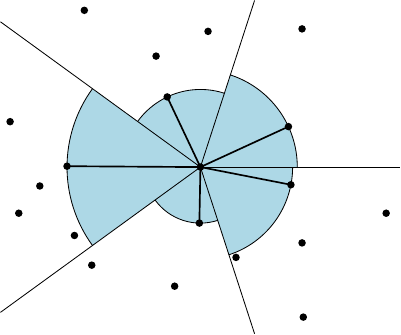}
 \caption{An example of the construction of the Yao graph with 5 cones.}
 \label{fig:YaoConstruction}
\end{figure}

One of the simplest ways of constructing a $t$-spanner is to first partition the plane around each vertex $x$ into a fixed number of cones\footnote{The orientation of the cones is the same for all vertices.} and then add edges connecting $x$ to a closest vertex in each cone (see Figure~\ref{fig:YaoConstruction}). Intuition suggests that this would yield a graph whose spanning ratio depends on the number of cones. Indeed, this is one of the first approximations of the complete Euclidean graph, referred to as \emph{Yao graphs} in the literature, introduced independently by Flinchbaugh and Jones~\cite{flinchbaugh1981strong} and Yao~\cite{yao1982constructing}. We denote the Yao graph by $Y_k$ where $k$ is the number of cones, each having angle $\theta= 2\pi/k$. Yao used these graphs to simplify computation of the Euclidean minimum spanning tree. Flinchbaugh and Jones studied their graph theoretic properties. Neither of them actually proved that they are $t$-spanners.

To the best of our knowledge, the first proof that Yao graphs are spanners was given by Alth{\"o}fer~\etal~\cite{althofer1993sparse}. They showed that for every $t > 1$, there exists a $k$ such that $Y_k$ is a $t$-spanner. It appears that some form of this result was known earlier, as Clarkson~\cite{clarkson1987approximation} already remarked in 1987 that $Y_{12}$ is a $1 + \sqrt{3}$-spanner, albeit without providing a proof or a reference. Bose~\etal~\cite{bose2004approximating} provided a more specific bound on the spanning ratio, by showing that for $k > 8$, $Y_k$ is a geometric spanner with spanning ratio at most $1 / (\cos \theta - \sin \theta)$. This was later strengthened to show that for $k>6$, $Y_k$ is a $1/(1-2\sin(\theta/2))$-spanner~\cite{bose2012piArxiv}. Damian and Raudonis~\cite{damian2012yao} showed that $Y_6$ is a $17.64$-spanner, and Bose~\etal~\cite{bose2012pi} showed that $Y_4$ is a $663$-spanner. For $k < 4$, El Molla~\cite{el2009yao} showed that there is no constant $t$ such that $Y_k$ is a $t$-spanner. This leaves open only the question of whether $Y_5$ is a constant spanner.

In this paper we close this gap by showing that for odd $k \geq 5$, the spanning ratio of $Y_k$ is at most $1/(1-2\sin(3\theta/8))$. This gives the first constant upper bound for $Y_5$ and implies that $Y_k$ is a constant spanner for all $k \geq 4$. For odd $k \geq 7$, our result also improves on the previous bound of $1/(1-2\sin(\theta/2))$. A more careful analysis allows us to reduce the upper bound on the spanning ratio of $Y_5$ from $10.9$ to $2+\sqrt{3} \approx 3.74$. We also give a lower bound of $2.87$ on the spanning ratio of $Y_5$. This complements a recent result on the spanning ratio of $\Theta_5$, which differs from $Y_5$ only in the distance measure it uses to select the closest neighbor in each cone: instead of Euclidean distance, it projects each vertex on the bisector of the cone and selects the vertex with the closest projection. Bose~\etal~\cite{bose2013theta5} showed that $\Theta_5$ has a spanning ratio in the interval $[3.79, 9.96]$. 
Because our upper bound of $3.74$ on the spanning ratio of $Y_5$ is slightly lower than the lower bound of $3.79$ on the spanning ratio of $\Theta_5$, this result establishes the first separation between the spanning ratio of Yao and $\Theta$-graphs. 
For all other $k \geq 4$, it is unclear which of $\Theta_k$ or $Y_k$ has a better spanning ratio.

In addition, we revisit the $Y_6$ graph, which plays a particularly important role as the transition between the graphs ($k > 6$) for which simple inductive proofs are known, and the graphs ($k \le 6$) whose best spanning ratios are established by complex arguments. Here we reduce the known spanning ratio of $Y_6$ from $17.64$ to $5.8$, thus moving toward the spanning ratio of $2$ established for $\Theta_6$ \cite{bonichon2010connections}. In contrast to $Y_5$, we present a lower bound of $2$ on the spanning ratio of $Y_6$, showing that it can never improve upon $\Theta_6$ in this regard.

Finally, we present the first lower bounds on the spanning ratio of Yao graphs with more than six cones, and a construction that shows that the Yao-Yao graph with five cones is not a spanner. The Yao-Yao graph is closely related to the Yao graph; a precise definition can be found in Section~\ref{sec:YaoYao5}.

Before delving into these problems, we introduce a few definitions common to all sections of this paper. In particular, we start with a more precise definition of the Yao graph $Y_k$. For a fixed $k$, let $Q_i(a)$ be the half-open cone of angle $2\pi/k$ with apex $a$, including the angle range $[ i, i+1) \cdot 2\pi/k$, for $i=0,\ldots,k-1$, where angles are measured counterclockwise from the positive $x$-axis. The directed graph $\overrightarrow{Y_k}$ includes exactly one directed edge from $a$ to a closest point in $Q_i(a)$, for each $i=0,\ldots,k-1$. If there are several equally-closest points within $Q_i(a)$, then ties are broken arbitrarily. The graph $Y_k$ is the undirected version of $\overrightarrow{Y_k}$. We use $\canon{a}{b}$ to denote the disk sector with center $a$ and radius $|ab|$ that subtends the cone with apex $a$ containing $b$. For any two points $a, b \in S$, we denote the length of a shortest path in $Y_k$ from $a$ to $b$ by $p(a,b)$.

\section{Spanning ratio of \texorpdfstring{$\boldsymbol{Y_k}$}{Yk}, for odd \texorpdfstring{$\boldsymbol{k}$}{k}}
\label{sec:yao5}

In this section we study the spanning properties of the Yao graphs $Y_k$ defined on a plane point set $S$ by an odd number of cones $k \ge 5$, each of angle $\theta = 2\pi/k$. For $k=5$ in particular, this is the first result showing that $Y_5$ is a constant spanner. For odd values $k > 5$, we improve the currently known bound on the spanning ratio of $Y_k$.

\begin{figure}[ht]
 \centering
 \includegraphics{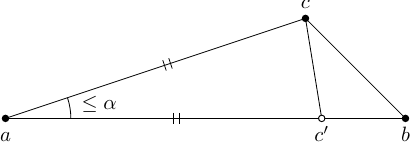}
 \caption{If $\alpha$ is small, there is a close relation between $|ac|$ and $|bc|$.}
 \label{fig:basicYaoLemma}
\end{figure}

\begin{lemma}
 \label{lem:basicyao}
 Given three points $a$, $b$, and $c$, such that $|ac| \leq |ab|$ and $\angle bac \leq \alpha < \pi$, then
 \[ |bc| \leq |ab| - \left( 1 - 2 \sin (\alpha/2) \right) |ac|. \]
\end{lemma}
\begin{proof}
 Let $c'$ be the point on $ab$ such that $|ac| = |ac'|$ (see~Figure~\ref{fig:basicYaoLemma}). Since $acc'$ forms an isosceles triangle,
  \[ |cc'| = 2 \sin (\angle bac / 2)  |ac| \leq 2 \sin (\alpha / 2) |ac|. \]
 Now, by the triangle inequality, 
 \begin{align*}
  |bc| &\leq |cc'| + |c'b|\\
       &\leq 2 \sin (\alpha / 2)  |ac| + |ab| - |ac'|\\
       &= |ab| - (1 - 2 \sin (\alpha / 2)) |ac|. \qedhere
 \end{align*}
\end{proof}

\begin{theorem}
 \label{thm:yaoodd}
 For any odd integer $k \geq 5$, the graph $Y_k$ has spanning ratio at most $t = 1 / (1 - 2 \sin(3\theta/8))$, where $\theta = 2\pi/k$.
\end{theorem}
\begin{proof} Let $a, b \in S$ be an arbitrary pair of points. We show that there is a path in $Y_k$ from $a$ to $b$ no longer than $t|ab|$. For simplicity, let $Q(a)$ denote the cone with apex $a$ that contains $b$, and let $Q(b)$ denote the cone with apex $b$ that contains $a$. Rotate the point set $S$ such that $Q(a)$ coincides with $Q_0(a)$, as depicted in Figure~\ref{fig:asymmetric}. We assume without loss of generality that $b$ lies below the bisector of $Q(a)$; the case when $b$ lies above this bisector is symmetric.

\begin{figure}[ht]
 \centering
 \includegraphics{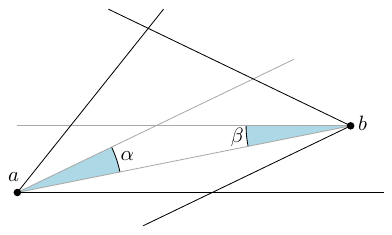}
 \caption{Since opposite cones are not symmetric, either $\alpha$ or $\beta$ is small.}
 \label{fig:asymmetric}
\end{figure}

Let $\alpha$ be the angle formed by the segment $ab$ with the bisector of $Q(a)$, and let $\beta$ be the angle formed by $ab$ with the bisector of $Q(b)$. Since $k$ is odd, the bisector of $Q(a)$ is parallel to the bottom boundary of $Q(b)$. Hence, we have that $\alpha = \theta/2 - \beta$. Assume without loss of generality that $\alpha$ is the smaller of these two angles (if not, we exchange the roles of $a$ and $b$). It follows that $\alpha \leq \theta/4$.

Our proof is by induction on the distance $|ab|$ (more formally, on the \emph{rank} of $\{a,b\}$ among all pairs of points when ordered by distance). 
In the base case $|ab|$ is minimal
, which means that there is no point $c \in Q(a)$ that is strictly closer to $a$ than $b$. Therefore either $ab \in Y_k$, in which case $p(a, b) = |ab|$ and our proof for the base case is finished, or there is a point $c \in Q(a)$ such that $|ab| = |ac|$ and $ac \in Y_k$. In this latter case, since $\alpha \leq \theta/4$ and $k \geq 5$, the angle between $ab$ and $ac$ is at most $\theta/2 + \alpha \leq 3\theta/4 \leq 3/4 \cdot (2\pi/5) = 3\pi/10$. This is less than $\pi/3$, which implies that $|bc| < |ab|$. This contradicts our assumption that $|ab|$ is minimal. It follows that 
$ab \in Y_k$ and the base case holds. 

For the inductive step, let $c \in Q(a)$ be such that $\overrightarrow{ac} \in \overrightarrow{Y_k}$. If $c$ coincides with $b$, then $p(a, b) = |ab|$ and the proof is finished. So assume that $c \neq b$. Because $c$ is the closest vertex to $a$ in this cone, and because $\angle cab \leq \theta/2 + \alpha \leq 3\theta/4$, we can apply Lemma~\ref{lem:basicyao} to derive $|cb| \leq |ab| - (1 - 2 \sin (3\theta/8))  |ac| = |ab| - |ac| / t$, which is strictly less than $|ab|$. Thus we can use the inductive hypothesis on $cb$ to determine a path between $a$ and $b$ of length
\[
  p(a, b) ~~\leq~~ |ac| + t |cb|
  ~~\leq~~ |ac| + t \left(|ab| - \frac{|ac|}{t}\right)
  ~~=~~ t  |ab|. \qedhere
\]
\end{proof}

\noindent
Applying this result to $Y_5$ yields a spanning ratio of $1 / (1 - 2 \sin(3\pi/20)) \approx 10.868$. This is the first known upper bound on the spanning ratio of $Y_5$ and fully settles the question of which Yao graphs are spanners.

\begin{corollary}
 The graph $Y_k$ is a spanner if and only if $k \geq 4$.
\end{corollary}

Next we lower the upper bound on the spanning ratio of $Y_5$ by taking a closer look at all feasible configurations.

\begin{theorem}
 The graph $Y_5$ has spanning ratio at most $2 + \sqrt{3} \approx 3.74$.
\end{theorem}

Here we also use induction on the pairwise distances between pairs of points in $S$. 
Consider the same configuration used in the proof of Theorem~\ref{thm:yaoodd}: $a \in Q(b)$ and $b \in Q(a)$ are points in $S$, and we seek a short path from $a$ and $b$. Without loss of generality, we assume that $Q(a)$ coincides with $Q_0(a)$, and $b$ lies below its bisector. 
If $|ab|$ is minimal (the base case) or $ab \in Y_5$, arguments similar to the ones used in the proof of Theorem~\ref{thm:yaoodd} show that $p(a, b) = |ab|$ and our proof is finished. So let $c \in Q(a)$ and $d \in Q(b)$ be the vertices in $S$ such that $\overrightarrow{ac} \in \overrightarrow{Y_5}$ and $\overrightarrow{bd} \in \overrightarrow{Y_5}$, and let $\phi = \angle cab$, and $\psi = \angle dba$ (see Figure~\ref{fig:yao5}a).

Now, instead of applying Lemma~\ref{lem:basicyao} for the maximum value of $\phi$ (as in the proof of Theorem~\ref{thm:yaoodd}), we apply Lemma~\ref{lem:basicyao} only for values $\phi \leq \tb$ or $\psi \leq \tb$, for some threshold angle $\tb$ (to be determined later).
These cases yield a spanning ratio of $t \geq 1 / (1 - 2 \sin(\tb/2))$. We handle the remaining cases differently, so for the remainder of the proof, we assume that $\phi > \tb$ and $\psi > \tb$.
We compute an exact value of $\tb$ shortly, but for now we only need that $\theta/2 < \tb < 3\theta/4$. This implies that neither $c$ nor $d$ can lie below $ab$, as this would make the corresponding angle smaller than $\theta/2$.

\begin{figure}[htb]
 \centering
 \begin{tabular}{cc}
  \includegraphics{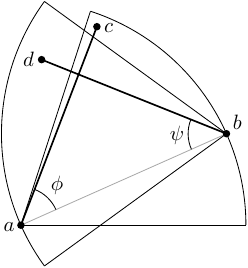} &
  \includegraphics{Yao5b} \\
  (a) & (b)
 \end{tabular} 
 \caption{(a) Two vertices of $Y_5$ with their closest vertices. (b) The worst-case situation when $ac$ and $bd$ cross.}
 \label{fig:yao5}
\end{figure}

First consider the case where $ac$ and $bd$ intersect. In this case, instead of directly applying an inductive argument to either $cb$ or $da$, we bound the distance $cd$ and use induction to show that $|ac| + t |cd| + |db| \leq t |ab|$.
To derive this bound, consider the point $c'$ such that $\angle c'ab = \tb$ and $|ac'| = |ab|$ and the analogously defined point $d'$ (see Figure~\ref{fig:yao5}b).
Let $s$ be the intersection point between $ac'$ and $bd'$.
When $ac$ and $bd$ intersect, the distance $|cd|$ can be increased by rotating $c$ towards $b$ and $d$ towards $a$.
Since both $\phi$ and $\psi$ must be larger than $\tb$, the worst case occurs when $\phi = \psi = \tb$, leaving $c$ and $d$ on the boundary of $\triangle c'd's$.
As $c'd'$ is the longest side of this triangle, it follows that $|cd| < |c'd'|$.
Using the fact that the triangles $\triangle c'd's$ and $\triangle abs$ are similar and isosceles, we can compute $|c'd'|$:

\begin{align*}
 |c'd'| &= 2 |c's| \cos \tb \\
        &= 2 (|ac'| - |as|) \cos \tb \\
        &= 2 \left(|ab| - \frac{|ab|}{2 \cos \tb}\right) \cos \tb \\
        &= (2 \cos \tb - 1) |ab|
\end{align*}

Recall that our aim is to use induction on $cd$ to obtain a short path from $a$ to $b$.
We now compute the spanning ratio $t$ required for the inequality $ |ac| + t  |cd| + |db| \leq t  |ab|$ to hold.
By the inequality above, we have that $|ac| + t  |cd| + |db| \leq |ab| + t  (2 \cos \tb - 1) |ab| + |ab|$.
This latter term is bounded above by $t |ab|$ for any $t \geq 1/(1 - \cos \tb)$.


So far we derived two constraints on $t$ and $\tb$: $t \geq 1 / (1 - 2 \sin(\tb/2))$ and $t \geq 1/(1 - \cos \tb)$.
Because $\sin \tb$ is increasing and $\cos \tb$ is decreasing for all values of $\tb$ under consideration, we minimize $t$ by choosing $\tb$ such that $1 / (1 - 2 \sin(\tb/2)) = 1/(1 - \cos \tb)$.
This yields $\tb = \arccos\big(\sqrt{3}-1\big) \approx 0.75$ and $t = 2 + \sqrt{3} \approx 3.74$.

Now consider what happens when one of $ac$ or $bd$ is ``short'', under some notion of short captured by the following lemma.

\begin{lemma}
 \label{lem:y5_short}
Let $\triangle abc$ be a triangle with angle $\alpha = \angle{cab}$ and longest side $ab$. Let $\lambda > 1$ be a real constant.
Then
 \[ 
  |ac| \leq \frac{2 \lambda^2 \cos \alpha - 2 \lambda}{\lambda^2 - 1}  |ab|~~~\text{ implies that }~~~|ac| + \lambda  |bc| \leq \lambda  |ab|.
 \]
\end{lemma}
\begin{proof}
First, note that the first inequality above implies $\lambda > 1/\cos \alpha$, as $|ac|$ would be non-positive otherwise.
By the law of cosines, 
 $|bc| = \sqrt{|ac|^2 + |ab|^2 - 2 |ab||ac| \cos \alpha}$.
By substituting this in the inequality 
 $|ac| + \lambda  |bc| \leq \lambda  |ab|$,
we see that it only holds if
$|ac| \leq \frac{2 \lambda^2 \cos \alpha - 2 \lambda}{\lambda^2 - 1} |ab|$, as stated by the lemma.
\end{proof}

\begin{figure}[htb]
 \centering
 \begin{tabular}{cc}
  \includegraphics{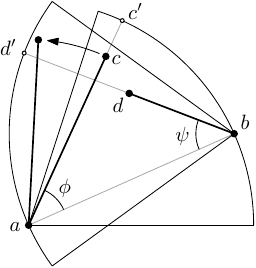} &
  \includegraphics{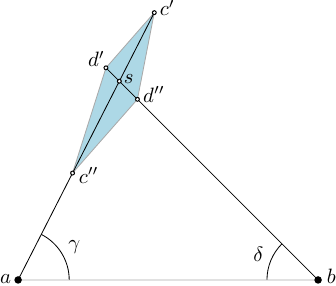} \\
  (a) & (b)
 \end{tabular}
 \caption{(a) The rotation to maximize $|cd|$ when $ac$ and $bd$ do not cross. (b) Illustration of Lemma~\ref{lem:yao5_fourpoints}}
 \label{fig:yao5_dist}
\end{figure}

The only case left to consider is when $ac$ and $bd$ are both long, but they do not intersect.
In this case, we again seek to  bound the distance $|cd|$.
If we can show that $|cd| \leq (2 \cos \tb - 1)  |ab|$, we can apply the same argument as for the intersecting case and we are done.
Let $c'$ be the point on the extension of $ac$ with $|ac'| = |ab|$, and let $d'$ be the analogous point on the extension of $bd$ (see Figure~\ref{fig:yao5_dist}a).
If $ac$ does not intersect $bd'$, we can rotate $d$ away from $c$ by increasing $\psi$.
Similarly, if $bd$ does not intersect $ac'$, we can rotate $c$ away from $d$ by increasing $\phi$.
Thus, the distance $|cd|$ is maximized when $\phi + \psi$ is maximal, which in our context happens when $\phi + \psi = 3\theta/2 = 3\pi/5$.
Note that in most cases, rotating this far moves the corresponding vertex past the boundary of the cone.
But since we are only trying to find an upper bound, this is not a problem.

Now let $c''$ be the point on the line through $ac$ with $|ac''| = \frac{2 t^2 \cos \phi - 2 t}{t^2 - 1} |ab|$, and let $d''$ be the point on the line through $bd$ with $|bd''| = \frac{2 t^2 \cos \psi - 2 t}{t^2 - 1} |ab|$.
If $c$ lies on $ac''$, Lemma~\ref{lem:y5_short} tells us that $|ac| + t |bc| \leq t |ab|$, which is exactly what we need.
The only difficulty is that the location of $c$ changed during the rotation.
But since the rotation preserved $|ac|$ and only increased $|bc|$, the inequality must hold for the configuration before the rotation as well.
The same argument applies for the case when $d$ lies on $bd''$.
The situation where $c$ and $d$ lie on $c''c'$ and $d''d'$, respectively, is handled by the following lemma.

\begin{lemma}\label{lem:yao5_fourpoints}
 Let $a, b, c, d \in S$.
 Let $\gamma = \angle cab$ and $\delta = \angle dba$ such that $\gamma > \tb$, $\delta > \tb$, and $\gamma + \delta = 3\pi/5$.
 Let $t = 2 + \sqrt{3}$. 
 If
\[
  \frac{2t^2 \cos \gamma - 2t}{t^2 - 1} |ab| \leq |ac| \leq |ab|
 ~\text{ and }~
  \frac{2t^2 \cos \delta - 2t}{t^2 - 1}  |ab| \leq |bd| \leq |ab|, 
\]
 then 
  $|cd| \leq (2 \cos \tb - 1)  |ab|$.
\end{lemma}
\begin{proof}
Assume without loss of generality that $\gamma \geq \delta$ and $|ab| = 1$.
Then $3\pi/10 \leq \gamma \leq 3\pi/5-\tb$ and $\tb \leq \delta \leq 3\pi/10$.
Let $c'$ be the point on the extension of $ac$ with $|ac'| = |ab|$, and let $d'$ be the analogous point on the extension of $bd$.
Let $s$ be the intersection of $ac'$ and $bd'$.
Let $c''$ be the point on the line through $ac$ with $|ac''| = \frac{2 t^2 \cos \gamma - 2 t}{t^2 - 1} |ab|$, and let $d''$ be the point on the line through $bd$ with $|bd''| = \frac{2 t^2 \cos \delta - 2 t}{t^2 - 1} |ab|$ (see Figure~\ref{fig:yao5_dist}b).
Let $c_1 = 2t^2/(t^2-1)$ and $c_2=1/\sin(3\pi/5)$.
We derive
\begin{align}
\frac{\mathrm{d} |ac''|}{\mathrm{d} \gamma} &= \frac{\mathrm{d} \left(\frac{2t^2\cos\gamma -2t}{t^2-1}\right)}{\mathrm{d} \gamma} = \frac{-2t^2\sin\gamma}{t^2-1} = -c_1\sin\gamma,\label{firstd}\\
\frac{\mathrm{d} |bd''|}{\mathrm{d} \gamma} &= \frac{\mathrm{d} \left(\frac{2t^2\cos\delta -2t}{t^2-1}\right)}{\mathrm{d} \gamma} = \frac{\mathrm{d} \left(\frac{2t^2\cos(3\pi/5-\gamma) -2t}{t^2-1}\right)}{\mathrm{d} \gamma} =\frac{2t^2\sin(3\pi/5-\gamma)}{t^2-1} = c_1\sin(3\pi/5-\gamma),\\
\frac{\mathrm{d} |as|}{\mathrm{d} \gamma} &= \frac{\mathrm{d} \left(\frac{\sin\delta}{\sin(\gamma+\delta)}\right)}{\mathrm{d} \gamma} = \frac{\mathrm{d} \left(\frac{\sin(3\pi/5-\gamma)}{\sin(3\pi/5)}\right)}{\mathrm{d} \gamma} =  \frac{-\cos(3\pi/5-\gamma)}{\sin(3\pi/5)} = -c_2\cos(3\pi/5-\gamma),\\
\frac{\mathrm{d} |bs|}{\mathrm{d} \gamma} &= \frac{\mathrm{d} \left(\frac{\sin\gamma}{\sin(\gamma+\delta)}\right)}{\mathrm{d} \gamma} = \frac{\mathrm{d} \left(\frac{\sin\gamma}{\sin(3\pi/5)}\right)}{\mathrm{d} \gamma} = \frac{\cos\gamma}{\sin(3\pi/5)} = c_2\cos\gamma.\label{lastd}
\end{align}
Let
\begin{align}
x_1 &= |as| - |ac''| = \frac{\sin\delta}{\sin(\gamma+\delta)} - \frac{2t^2\cos\gamma -2t}{t^2-1}, \label{x1}\\
x_2 &= |ac'| - |as| = 1- \frac{\sin\delta}{\sin(\gamma+\delta)},\label{x2}\\
y_1 &= |bs| - |bd''| = \frac{\sin\gamma}{\sin(\gamma+\delta)} - \frac{2t^2\cos\delta -2t}{t^2-1},\label{y1}\\
y_2 &= |bd'| - |bs| = 1 - \frac{\sin\gamma}{\sin(\gamma+\delta)}\label{y2}.
\end{align}
Note that the values of $x_1$ and $y_1$ could be negative if $c''$ or $d''$ lie past $s$. Substituting $c_1$, $c_2$, and (\ref{firstd}) - (\ref{lastd}) in the equalities above yields
\begin{align}
\frac{\mathrm{d} x_1}{\mathrm{d} \gamma} &= \frac{\mathrm{d}(|as| - |ac''|)}{\mathrm{d} \gamma} = -c_2\cos(3\pi/5-\gamma) + c_1\sin\gamma,\label{d1}\\
\frac{\mathrm{d} x_2}{\mathrm{d} \gamma} &= \frac{\mathrm{d}(|ac'| - |as|)}{\mathrm{d} \gamma} = \frac{\mathrm{d}(1 - |as|)}{\mathrm{d} \gamma} = c_2\cos(3\pi/5-\gamma),\\
\frac{\mathrm{d} y_1}{\mathrm{d} \gamma} &= \frac{\mathrm{d}(|bs| - |bd''|)}{\mathrm{d} \gamma} = c_2\cos\gamma - c_1\sin(3\pi/5-\gamma),\\
\frac{\mathrm{d} y_2}{\mathrm{d} \gamma} &= \frac{\mathrm{d}(|bd'| - |bs|)}{\mathrm{d} \gamma} = \frac{\mathrm{d}(1 - |bs|)}{\mathrm{d} \gamma} = -c_2\cos\gamma.\label{d4}
\end{align}
Recall that 
$c_1 = 2t^2/(t^2-1)$, $c_2=1/\sin(3\pi/5)$, 
and $3\pi/10 \leq \gamma \leq 3\pi/5-\tb$.
We verify the following:
\begin{align*}
\frac{\mathrm{d}^2 x_1}{\mathrm{d} \gamma^2} &= -c_2\sin(3\pi/5-\gamma)+c_1\cos\gamma \\
                                                                            & > - 1.1 \sin(3\pi/10)+2.1 \cos(3\pi/5-\tb)  > 0,\\
\frac{\mathrm{d}^2 x_2}{\mathrm{d} \gamma^2} &= c_2\sin(3\pi/5-\gamma) > 0,\\
\frac{\mathrm{d}^2 y_1}{\mathrm{d} \gamma^2} &= -c_2\sin\gamma+c_1\cos(3\pi/5-\gamma) \\
                                                                            & > - 1.1 \sin(3\pi/5-\tb)+2.1 \cos(3\pi/10)  > 0,\\
\frac{\mathrm{d}^2 y_2}{\mathrm{d} \gamma^2} &= c_2\sin\gamma > 0.
\end{align*}
Therefore, by substituting $\gamma = 3\pi/10$ or $\gamma = 3\pi/5-\tb$ as the lower- or upper-bound of $\gamma$ into (\ref{d1}) - (\ref{d4}), we can verify the following ranges:
\begin{align}
-c_2\cos(3\pi/10) + c_1\sin(3\pi/10) \leq & \frac{\mathrm{d} x_1}{\mathrm{d} \gamma} \leq -c_2\cos\tb + c_1\sin(3\pi/5-\tb), \nonumber \\
c_2\cos(3\pi/10) \leq & \frac{\mathrm{d} x_2}{\mathrm{d} \gamma} \leq c_2\cos\tb,\label{dx2}\\
c_2\cos(3\pi/10) - c_1\sin(3\pi/10) \leq & \frac{\mathrm{d} y_1}{\mathrm{d} \gamma} \leq c_2\cos(3\pi/5-\tb)-c_1\sin\tb, \nonumber \\
-c_2\cos(3\pi/10) \leq & \frac{\mathrm{d} y_2}{\mathrm{d} \gamma} \leq -c_2\cos(3\pi/5-\tb). \nonumber
\end{align}
Specifically, we can verify that
\begin{align}
\frac{\mathrm{d} x_1}{\mathrm{d} \gamma} \geq \max\left(\frac{\mathrm{d} x_2}{\mathrm{d} \gamma},\left|\frac{\mathrm{d} y_1}{\mathrm{d} \gamma}\right|,\left|\frac{\mathrm{d} y_2}{\mathrm{d} \gamma}\right|\right),\label{compared}
\end{align}
which implies
 \[ \frac{\mathrm{d} (x_1-x_2)}{\mathrm{d} \gamma} = \frac{\mathrm{d} x_1}{\mathrm{d} \gamma}-\frac{\mathrm{d} x_2}{\mathrm{d} \gamma} > 0. \]
By simply plugging in $\gamma = 3\pi/10$ into (\ref{x1}) and (\ref{x2}), we verify that $(x_1-x_2) > 0$ when $\gamma = 3\pi/10$ and hence $x_1 > x_2$ for all $\gamma \in [3\pi/10, 3\pi/5-\tb]$.
Similarly, we have $x_2 > 0$ when $\gamma = 3\pi/10$, and hence by (\ref{dx2}), $x_2 > 0$ for all $\gamma \in [3\pi/10, 3\pi/5-\tb]$. 
These together yield $x_1 > x_2 > 0$.
By the triangle inequality,
\begin{align*}
|c''d''| &\leq |sc''|+|sd''| = |x_1| + |y_1| = x_1 + |y_1|, \\
|c''d'| &\leq |sc''|+|sd'| = |x_1| + |y_2| = x_1 + |y_2|,\\
|c'd''| &\leq |sc'|+|sd''| = |x_2| + |y_1| \leq x_1 + |y_1|,\\
|c'd'| &\leq |sc'|+|sd'| = |x_2| + |y_2| \leq x_1 + |y_2|.
\end{align*}
By (\ref{compared}),
\begin{align*}
\frac{\mathrm{d} (x_1 + |y_1|)}{\mathrm{d} \gamma} & \geq \frac{\mathrm{d} (x_1)}{\mathrm{d} \gamma} - \left|\frac{\mathrm{d} (y_1)}{\mathrm{d} \gamma}\right| \geq 0,\\
\frac{\mathrm{d} (x_1 + |y_2|)}{\mathrm{d} \gamma} & \geq \frac{\mathrm{d} (x_1)}{\mathrm{d} \gamma} - \left|\frac{\mathrm{d} (y_2)}{\mathrm{d} \gamma}\right| \geq 0.
\end{align*}
By substituting $\gamma = 3\pi/5-\tb$ into (\ref{x1}), (\ref{y1}), and (\ref{y2}), one can easily verify that $x_1 + |y_1| \leq 2\cos\tb-1$ and $x_1 + |y_2| \leq 2\cos\tb-1$ when $\gamma$ is maximized.
Therefore $\max(x_1 + |y_1|, x_1 + |y_2|) \leq 2\cos\tb-1$ for all $\gamma \in [3\pi/10, 3\pi/5-\tb]$, and hence $|cd| \leq \max(|c''d''|,|c''d'|,|c'd''|,|c'd'|) \leq 2\cos\tb-1 $ as required.
\end{proof}

This completes the proof for the upper bound. 
Next, we prove a lower bound on the spanning ratio.

\begin{theorem}
$Y_5$ has spanning ratio at least $2.87$. 
\end{theorem}
\begin{proof}
The inductive proof of the upper bound on the spanning ratio of $Y_5$ suggests a construction for a lower bound.
It is based on recursively attaching the ``lattice cross'' shown in Figure~\ref{fig:yao5}b to pairs of non-adjacent points (e.g., pairs $\{a,d'\}, \{b,c'\}, \{c',d'\}$ in Figure~\ref{fig:yao5}b).
This recursive construction results in a fractal-like shape, starting from the pair $\{a,b\}$ (see Figure~\ref{fig:frac-4a}).
However, the growth of the fractal is limited by collisions of neighboring fractal branches that create shortcuts to the paths, as shown in the circled area of Figure~\ref{fig:frac-4a}.
This construction yields a spanning ratio of 2.66.
 
\begin{figure}[htbp]
 \centering
 \includegraphics{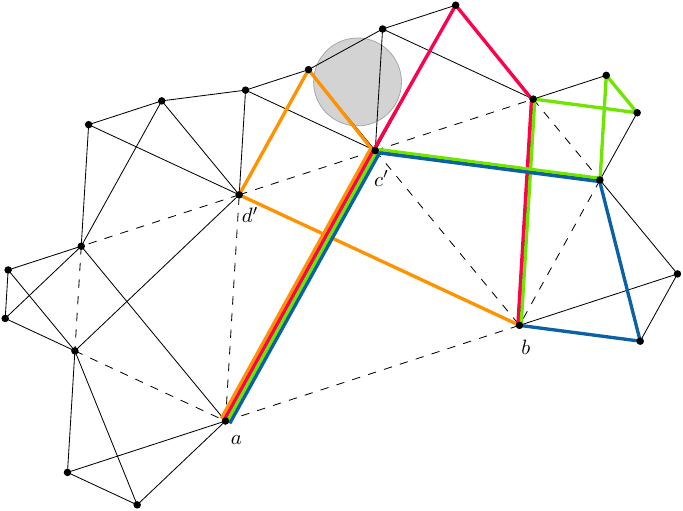}
 \caption{Spanning ratio $2.66$. The fractal growth is limited by collision of branches in the circled area. The shortest paths between $a$ and $b$ are colored.}
 \label{fig:frac-4a}
\end{figure}

\begin{figure}[htbp]
 \centering
 \includegraphics{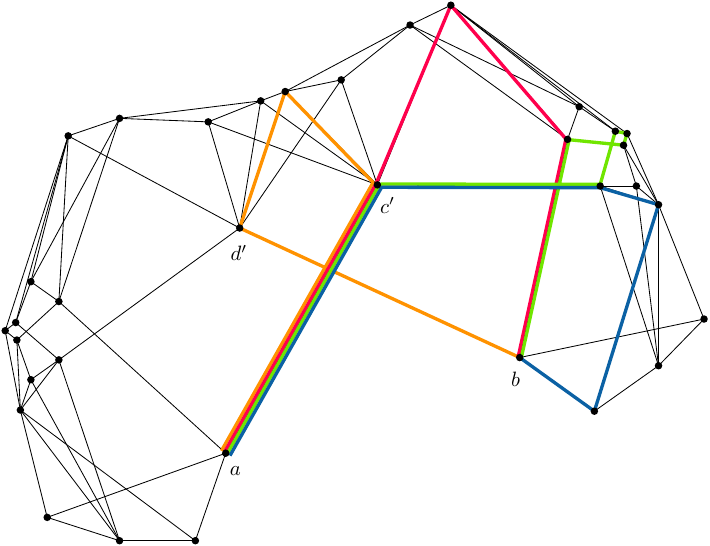}
 \caption{The spanning ratio is increased to 2.87 by adjusting the shape to equalize the lengths of the shortest paths between $a$ and $b$.}
 \label{fig:frac-4b}
\end{figure}



We adjust the shape of the fractal to increase the spanning ratio.
In Figure~\ref{fig:frac-4b}, we obtain a spanning ratio of more than 2.87 by equalizing the length of all shortest paths between \mbox{$a$ and $b$}.
The coordinates for the points in Figure~\ref{fig:frac-4b} can be found in Appendix~\ref{app:LB-coordinates}.
\end{proof}

\section{Spanning ratio of \texorpdfstring{$\boldsymbol{Y_6}$}{Y6}}
\label{sec:yao6}

In this section we fix $k = 6$ and show that, for any pair of points $a, b \in S$,  $p(a, b) \le 5.8 |ab|$.
We also establish a lower bound of $2$ for the spanning ratio of $Y_6$. 
Our proof is inductive and it relies on two simple lemmas, which we introduce next.

Let $a, b \in S$ and let $\overrightarrow{ac} \in \overrightarrow{Y_6}$ be the edge from $a$ within the cone that includes $b$.
The next two lemmas will be relevant in the context where we seek to bound $p(a,b)$ by applying the induction hypothesis to $p(c, b)$. 
The basic geometry is illustrated in Figure~\ref{fig:TriangleNotation}.

\begin{figure}[htbp]
\centering
\includegraphics{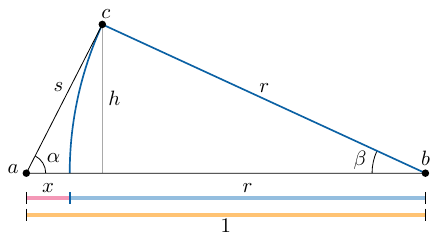}
\caption{Notation for triangle $\triangle abc$. The dimensions have been normalized so that $|ab|=1$.}
\label{fig:TriangleNotation}
\end{figure}

\begin{lemma}[Triangle]
Let $\triangle abc$ be labeled as in Figure~\ref{fig:TriangleNotation},
with $|ac| \le |ab|$, $|bc| < |ab|$, $x=|ab|-|bc|$ and $s=|ac|$.
The ratio $s/x$ is equal to some function $t$ that depends on $\alpha$ and $\beta$:
\begin{equation}
    \frac{s}{x} = t(\alpha, \beta) = \frac{\cos(\beta/2)}{\cos(\alpha+\beta/2)}.
\label{eq:ab}
\end{equation}
\label{lem:AlphaBeta}
\end{lemma}
\begin{proof}
Normalize the triangle so that $|ab|=1$;
this does not alter the quantity we seek to compute,  $s/x$.
Let $|bc|=r$ to simplify notation.
Then $x=1-r$ and $x \ge 0$  because $r = |bc| \le |ab| = 1$. 
Note that each of the angles $\angle cab$ and $\angle cba$ is strictly less than 
$\pi/2$, because $|ac| \le |ab|$ and $|bc| \le |ab|$. Thus the projection 
of $c$ onto $ab$ is interior to the segment $ab$. 
Computing the altitude $h$ of $\triangle abc$ in two ways yields
$$
s \sin \alpha  = r \sin \beta.
$$
Also projections onto $ab$ yield
$$
s \cos \alpha  + r \cos  \beta  = 1.
$$
Solving these two equations simultaneously yields expressions for $r$
and $s$ as functions of $\alpha$ and $\beta$:
$$
r = \frac { \sin \alpha } { \sin \alpha \cos \beta + \cos \alpha \sin \beta }, \quad\;
s = \frac { \sin \beta } { \sin \alpha \cos \beta + \cos \alpha \sin \beta }
$$
Now we can compute $s/x = s/(1-r)$ as a function of $\alpha$ and $\beta$.
This simplifies to
$$
\frac{s}{x} = \frac{ \cos ( \beta/2 ) } { \cos ( \alpha + \beta/2 ) } 
$$
as claimed. 
\end{proof}

The following lemma derives an upper bound on the function $t(\alpha, \beta)$ from Lemma~\ref{lem:AlphaBeta}, 
which will be used in Theorem~\ref{thm:Y6spanner} to derive an optimal value for $\delta$.

\begin{lemma}
Let $a, b, c \in S$ satisfy the conditions of Lemma~\ref{lem:AlphaBeta}, 
and let $t(\alpha, \beta)$ be as defined in~(\ref{eq:ab}).  
Let $\delta \in (0, \pi/3)$ be a fixed positive angle. 
If $\alpha \le  \pi/3 - \delta$, or $\beta \le  \pi/3 -\delta$, then
\[ t(\alpha,\beta) \le t(\pi/3, \pi/3-\delta) =   \frac{\cos(\pi/6-\delta/2)}{\sin(\delta/2)}.
\]
\label{lem:AlphaBetaMax}
\end{lemma}
\begin{proof}
The derivative of $t(\alpha, \beta)$ with respect to $\alpha$ is 
\[
    \frac{\partial t}{\partial \alpha} =  \frac{\sin \alpha+\sin(\alpha+\beta)}{1+\cos(2\alpha+\beta)} > 0.
\]
This means that, for a fixed $\beta$ value, $t(\alpha, \beta)$ reaches its maximum when 
$\alpha$ is maximum. Similarly, the derivative of $t(\alpha, \beta)$ with respect to $\beta$ is 
\[
    \frac{\partial t}{\partial \beta} =  \frac{\sin \alpha}{2\cos(\alpha+\beta/2)^2} > 0.
\]
So for a fixed value $\alpha$ value, $t(\alpha, \beta)$ reaches its maximum when 
$\beta$ is maximum. Because $|ac| \le |ab|$, $\beta \le \angle{acb}$. The sum of these two angles is 
$\pi-\alpha$, therefore $\beta \le \pi/2-\alpha/2$. This along with the derivations above implies 
that, for a fixed value $\alpha \le \pi/3-\delta$, $t(\alpha, \beta) \le t(\alpha,\pi/2-\alpha/2) \le t(\pi/3-\delta, \pi/3+\delta/2)$ (we substituted $\alpha = \pi/3 -\delta$ in this latter inequality).
Next we evaluate
\[
  \frac{t(\pi/3, \pi/3-\delta)}{t(\pi/3-\delta, \pi/3+\delta/2)} = 
  \frac{\cos(\pi/6-\delta/2) \sin(3\delta/2)}{\cos(\pi/6+\delta/2)\sin(\delta/2)} > 1.
\]
It follows that $t(\pi/3, \pi/3-\delta)$ is maximal. 
\end{proof}

%
%
We are now ready to prove the main result of this section.

\begin{theorem}
 \label{thm:Y6spanner}
The graph $Y_6$ has spanning ratio at most $5.8$. 
\end{theorem}
This result follows from the following lemma, with the variable $\delta$ substituted by the quantity $\delta_0 = 0.324$ that minimizes $t(\delta)$. (It can be easily verified that $t(\delta) \ge t(0.324)$ and $ t(0.324) < 5.8$.)

\begin{lemma}
Let $\delta  \in (0, \pi/9)$ be a strictly positive real value. The graph $Y_6$ has spanning ratio bounded above by  
\begin{equation}
t = t(\delta) = \max\left\{\frac{\cos(\pi/6-\delta/2)}{\sin(\delta/2)}, ~\frac{2}{1-\frac{\sin(2\delta)}{\sin(\pi/6+2\delta)}}\right\}.
\label{eq:t}
\end{equation}
\label{lem:main}
\end{lemma}
\begin{proof}
The proof is by induction on the pairwise distance between pairs of points $a, b \in S$. 
Without loss of generality let $b \in Q_0(a)$. 

\paragraph{Base case} We show that, if $|ab|$ is minimal, then 
$\overrightarrow{ab} \in \overrightarrow{Y_6}$ and so $p(a,b) = |ab|$.
%
If $\overrightarrow{ab} \in \overrightarrow{Y_6}$, 
then the lemma holds.  So assume that $\overrightarrow{ab} \not\in \overrightarrow{Y_6}$; 
we will derive a contradiction. Because $\overrightarrow{ab} \not\in \overrightarrow{Y_6}$, there must be 
another point $c \in Q_0(a)$ such that $\overrightarrow{ac} \in \overrightarrow{Y_6}$ and 
$|ac|=|ab|$. 
Let $\alpha_1$ and $\alpha_2$ be the angles that $ab$ and $ac$ make with the horizontal respectively.
Because both $\alpha_1,\alpha_2 \in [0,\pi/3)$, necessarily $|\alpha_1 - \alpha_2| < \pi/3$.
Thus $|bc| < |ab|=|ac|$, contradicting the assumption that $|ab|$ is minimal.  
So in fact it must be that $\overrightarrow{ab} \in \overrightarrow{Y_6}$, and the lemma is established.

\paragraph{Main idea of the inductive step} 
It has already been established that $Y_7$ is a spanner~\cite{bose2012piArxiv}; 
the sector angles for $Y_7$ are $2\pi/7$.
The main idea of our inductive proof is to partition the 
$\pi/3$-sectors of $Y_6$ into peripheral cones of angle $\delta$, for some 
fixed $\delta \in (0,\pi/9)$, leaving a central sector of angle $\pi/3 - 2\delta$.
(The $\delta$-cones are the shaded regions in Figure~\ref{fig:Y6SectorDelta1}.) 


When an edge of $Y_6$ falls inside the central sector, induction will apply, 
because an edge within the central sector makes definite progress toward
the goal in that sector (as it does in $Y_7$), ensuring that the remaining distance to be covered is
strictly smaller than the original. This idea is captured by the following lemma.
\begin{lemma}[Induction Step]
Let $a, b, c \in S$ such that $b$ and $c$ lie in the same cone with apex $a$, and 
$\overrightarrow{ac} \in \overrightarrow{Y_6}$. Let $\alpha = \angle cab$ and $\beta = \angle cba$. 
If either $\alpha <  \pi/3 - \delta$ or $\beta <  \pi/3 -\delta$, then 
we may use induction on $p(c,b)$ to conclude that $p(a,b) \le t |ab|$.
\label{lem:Induct}
\end{lemma}
\begin{proof}
This configuration is depicted in Figure~\ref{fig:TriangleNotation}.
Because $\overrightarrow{ac} \in \overrightarrow{Y_6}$ and $b$ and $c$ lie in the same cone 
with apex $a$, we have that $|ac| \le |ab|$. 
Because at least one of $\alpha$ or $\beta$ is strictly smaller than $\pi/3$, 
we have that $|cb| < |ab|$. 
Thus the conditions of Lemma~\ref{lem:AlphaBeta} are satisfied, so we can 
use Lemma~\ref{lem:AlphaBeta} to bound $|ac|$ in terms of $x = |ab|-|bc|$: since $|ac|/x < t$, $|ac| < t x$.
Because $|cb| < |ab|$, we may apply induction to bound $p(c,b)$:  
$p(c, b) \le t |cb|$. Hence 
\[
 p(a,b)  \le  |ac| + p(c,b) \le  t x + t |cb|  =  t( x + |cb| ) =  t |ab|. \qedhere
\]
\end{proof}
We will henceforth use the symbol \fbox{\emph{Induct}} as shorthand for applying 
Lemma~\ref{lem:Induct} to a triangle equivalent to that in
Figure~\ref{fig:TriangleNotation}. 

Lemma~\ref{lem:Induct} leaves out $Y_6$ edges falling within the $\delta$-cones, that 
could conceivably \emph{not} make progress 
toward the goal. For example, following one edge of an equilateral triangle leaves one 
exactly as far away from the other corner as at the start.
However, we will see that when all relevant edges of $Y_6$ fall within the $\delta$-cones
, the restricted geometric structure ensures that progress toward the goal is 
indeed made, and again induction applies.

\paragraph{Inductive step}  
The inductive step proof first handles the cases where edges of $Y_6$ directed from $a$ or from $b$  fall in the central 
portion of the relevant sectors, and so satisfy Lemma~\ref{lem:AlphaBetaMax}, and so 
Lemma~\ref{lem:Induct} applies.

Recall that $b \in Q_0(a)$ by our assumption. 
If $\overrightarrow{ab} \in \overrightarrow{Y_6}$, 
then $p(a,b)=|ab|$ and we are finished. Assuming otherwise, there must be a point $c \in Q_0(a)$ such that  
$\overrightarrow{ac} \in \overrightarrow{Y_6}$ and $|ac| \le |ab|$. For the remainder of the proof, we are in 
this situation
.
The proof now partitions into three parts: 
(1)~when only $Q_0(a)$ is relevant and leads to \fbox{\emph{Induct}};
(2)~when $Q_2(b)$ leads to \fbox{\emph{Induct}};
(3)~when we fall into a special situation, for which induction also applies, but for different reasons.

\paragraph{(1)~The {\large $\boldsymbol{Q_0(a)}$} sector}
Consider $\triangle abc$ as previously illustrated in Figure~\ref{fig:TriangleNotation}.
If either $b$ or $c$ is not in one of the $\delta$-cones of $Q_0(a)$, then 
$\alpha = \angle bac < \pi/3 - \delta$: \fbox{\emph{Induct}.}

%
\begin{figure}[htb]
    \centering
    \begin{tabular}{ccc}
        \includegraphics{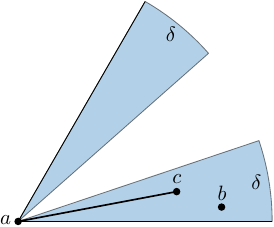} & 
        \includegraphics{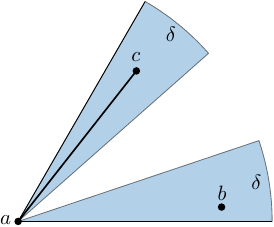} & 
        \includegraphics{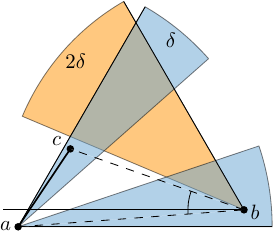} \\
        (a) & (b) & (c) 
    \end{tabular}
    \caption{Different cases for the location of $c$; the $\delta$-cones are shaded. (a)~$b$ and $c$ in the same $\delta$-cone. (b)~$b$ and $c$ in different $\delta$-cones. (c)~$c$ in the upper $\delta$-cone of $Q_0(a)$, but not in the upper $2\delta$-cone of $Q_2(b)$; here 
$\angle abc$ is ``small''.}
    \label{fig:Y6SectorDelta1}
\end{figure}
%

Now assume that both $b$ and $c$ lie in $\delta$-cones of $Q_0(a)$. 
If they both lie within the same $\delta$-cone (Figure~\ref{fig:Y6SectorDelta1}a),
then again $\alpha$ is small: \fbox{\emph{Induct}.} 
So without loss of generality let $b$ lie in the lower $\delta$-cone, and $c$ in the upper
$\delta$-cone of $Q_0(a)$; see Figure~\ref{fig:Y6SectorDelta1}b. 
We cannot apply induction in this situation because the 
ratio $s/x$ in Lemma~\ref{lem:AlphaBeta} has no upper bound.

\paragraph{(2)~The {\large $\boldsymbol{Q_2(b)}$} sector}
Now we consider $Q_2(b)$, the sector with apex at $b$ aiming to the left of $b$, 
and assume that $c \in Q_2(b)$. Refer to Figure~\ref{fig:Y6SectorDelta1}c. 
The case $c \notin Q_2(b)$ will be discussed later (special situation).

Because $b$ may subtend an angle as large as $\delta$ at $a$ with the horizontal,
the ``upper $2\delta$-cone'' of $Q_2(b)$ becomes the relevant region.
If $c$ is not in the upper $2\delta$-cone of $Q_2(b)$ (as depicted in Figure~\ref{fig:Y6SectorDelta1}c),
then $\triangle abc$ satisfies Lemma~\ref{lem:AlphaBetaMax} with 
$\angle abc < \pi/3 - \delta$: 
\fbox{\emph{Induct}.}
Note that this conclusion follows even if $c$ is in the small region
outside of and below $Q_2(b)$: the angle 
$\angle abc$ at $b$ is then very small.

Assume now that $c$ is in the upper $2\delta$-cone of $Q_2(b)$. 
Let $d \in Q_2(b)$ be the point such that  $\overrightarrow{bd} \in \overrightarrow{Y_6}$.
We now consider possible locations for $d$.
If $d = c$, then $p(a,b) \le |ac| + |cb| \le 2 |ab|$, and we are finished.
So assume henceforth that $d$ is distinct from $c$.


\begin{figure}[hpt]
    \centering
	\begin{tabular}{ccc}
        \includegraphics{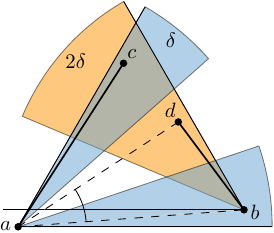} & 
        \includegraphics{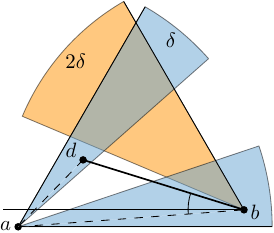} & 
        \includegraphics{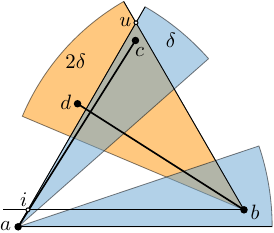} \\
        (a) & (b) & (c) 
    \end{tabular}
    \caption{Different cases for the location of $d$. (a)~$d$ not in the upper $\delta$-cone of $Q_0(a)$: $\angle bad$ is small. (b)~$d$ not in the upper $2\delta$-cone of $Q_2(b)$: $\angle abd$ is small. (c) $d$ in the upper $2\delta$-cone of $Q_2(b)$, but not necessarily in the upper $\delta$-cone of $Q_0(a)$: $|cd| < |ab|$~(Lemma~\protect\ref{lem:cd-close}).}
    \label{fig:Pointd}
\end{figure}

If $d$ is not in the upper $\delta$-cone of $Q_0(a)$ (Figure~\ref{fig:Pointd}a), then
$\triangle abd$ satisfies  Lemma~\ref{lem:AlphaBetaMax} with the roles of $a$ and $b$ 
reversed: $bd$ takes a step toward $a$, with the angle at $a$ satisfying 
$\angle bad < \pi/3 - \delta$: \fbox{\emph{Induct}.}

If $d$ is not in the upper $2\delta$-cone of $Q_2(b)$ (Figure~\ref{fig:Pointd}b), then
$\triangle abd$ satisfies Lemma~\ref{lem:AlphaBetaMax} again with
the roles of $a$ and $b$ reversed and this time the angle at $b$ bounded away
from $\pi/3$, $\angle abd < \pi/3 - \delta$: \fbox{\emph{Induct}.}

Assume now that $d$ is in the upper $2\delta$-cone of $Q_2(b)$.  
Recall that we are in the situation where $c$ 
also lies in the upper $2\delta$-cone of $Q_2(b)$, so it is close to $d$.
(Note however that $c$ and $d$ may lie on either the same side, or on opposite sides of the upper ray of $Q_0(a)$.) 
See Figure~\ref{fig:Pointd}c.
This suggests the strategy of following $ac$ and $db$, connected by $p(c, d)$.
We show that in fact $|cd| < |ab|$, so the inductive hypothesis can be applied to $p(c,d)$. 
More precisely, we show the following result.

\begin{lemma}
Let $a, b, c, d \in S$ be as in Figure~\ref{fig:Pointd}c, with  
$\overrightarrow{bd} \in \overrightarrow{Y_6}$, $b, c \in Q_0(a)$ and 
$c, d \in Q_2(b)$. If both $c$ and $d$ lie above the lower rays bounding
the upper 
$\delta$-cone of $Q_0(a)$ and the upper $2\delta$-cone of $Q_2(b)$, then for any $0 \le \delta \le \pi/9$,
\begin{equation}
 |cd| \le \frac{\sin(2\delta)}{\sin(\pi/6+2\delta)}|ab|.
\label{eq:cd}
\end{equation}
\label{lem:cd-close}
\end{lemma}
%
Note that $c$ lies in the intersection region between the upper $\delta$-cone of $Q_0(a)$ and the upper 
$2\delta$-cone of $Q_2(b)$, because $c \in Q_0(a) \cap Q_2(b)$ (by the statement of the lemma). However, Lemma~\ref{lem:cd-close} 
does not restrict the location of $d$ to the same region. Indeed, $d$ may lie either below or above the upper ray 
bounding $Q_0(a)$, as long as it satisfies the condition $|bd| \le |bc|$. (This condition must hold because $c, d$ are in the same 
sector $Q_2(b)$, and $\overrightarrow{bd} \in \overrightarrow{Y_6}$.) 
To keep the flow of our main proof uninterrupted, we defer a proof of Lemma~\ref{lem:cd-close} to Section~\ref{sec:cd-close}.


By Lemma~\ref{lem:cd-close} we have $|cd| < |ab|$. Thus we can use the induction hypothesis to 
show that $p(c, d) \le t|cd|$. 
We know that $|ac| \le |ab|$, because both $b$ and $c$ are in $Q_0(a)$ and
$\overrightarrow{ac} \in \overrightarrow{Y_6}$.
We also know that $|bd| \le |bc|$ because both $c$ and $d$ are in $Q_2(b)$
and $\overrightarrow{bd} \in \overrightarrow{Y_6}$.
Let $u$ and $i$ be the upper and lower intersection points between the rays bounding $Q_2(b)$ and 
the upper ray of $Q_0(a)$, as in Figure~\ref{fig:Pointd}c. Note that $\triangle bui$ is equilateral,
and because $c$ lies in this triangle, we have $|bc| \le |bu| = |bi| \le |ab|$. It follows that 
$|bd| \le |ab|$. 
So in this situation (illustrated in Figure~\ref{fig:Pointd}c), we have:
\begin{align*}
p(a,b) &\le |ac| + p(c,d) + |bd| \\
       &\le 2 |ab| + p(c,d) \\
       &\le 2 |ab| + t |cd| \\
       &\le  2 |ab| + t \frac{\sin(2\delta)}{\sin(\pi/6+2\delta)}|ab| \\
       &\le  t |ab|.
\end{align*}
Here we have applied Lemma~\ref{lem:cd-close} to bound $|cd|$. Note that the latter inequality above is true for the 
value of $t$ from~(\ref{eq:t}).

\newcommand{\nc}{\ensuremath{z}}

\paragraph{(3)~Special situation}
The only case left to discuss is the one in which $c$ lies in the upper $\delta$-cone of $Q_0(a)$ 
and to the right of the upper ray of  $Q_2(b)$. 
This situation is depicted in Figure~\ref{fig:yao6case0}.  
Next consider $Q_4(c)$. Because $b \in Q_4(c)$, there exists $\overrightarrow{c\nc} \in \overrightarrow{Y_6}$, 
with $\nc \in Q_4(c)$ and $|c\nc| \le |cb|$. Clearly $\nc \in Q_0(a) \cup Q_5(a)$. 
Note that the disk sector $\canon{a}{c} \subset Q_0(a)$ with center $a$ and radius $|ac|$ must be empty, 
because $\overrightarrow{ac} \in \overrightarrow{Y_6}$. 


\begin{figure}[htbp]
    \centering
    \begin{tabular}{c@{\hspace{0.1\linewidth}}c}
        \includegraphics{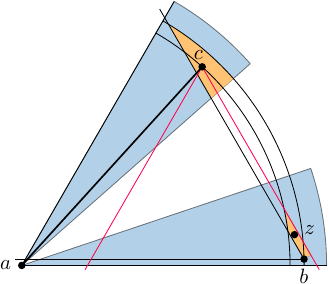} &
        \includegraphics{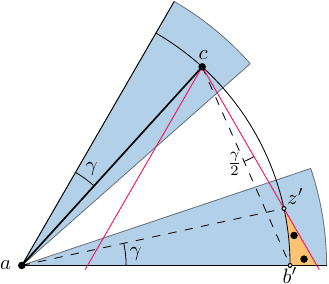} \\
        (a) & (b) 
    \end{tabular}
    \caption{(a) The special situation, with $c \notin Q_2(b)$ and $\nc \in Q_0(a)$. (b) Both $b$ and $z$ must lie in the highlighted region, so we can apply Lemma~\ref{lem:cd-close}.}
    \label{fig:yao6case0}
\end{figure}

\vspace{-1em}
\paragraph{Case 3(a)} If $\nc \in Q_0(a)$, then $\nc$ lies in the lower $\delta$-cone of $Q_0(a)$ and to 
the right of $\canon{a}{c}$, close to $b$. 
See Figure~\ref{fig:yao6case0}a. 
In this case we show that the quantity on the right side of inequality~(\ref{eq:cd}) 
is a loose upper bound on $|b\nc|$, and that similar inductive arguments hold here as well. 
Let the circumference of $\canon{a}{c}$ intersect the right ray of $Q_4(c)$ and the 
lower ray of $Q_0(a)$ at points $\nc' \neq c$ and $b'$, respectively. Refer to Figure~\ref{fig:yao6case0}b.
Let $\gamma \le \delta$ be the angle formed by $ac$ with the upper ray of $Q_0(a)$. 
Then $\angle \nc'ab' = \gamma$ and $\angle \nc'cb' = \gamma/2$ by the inscribed angle theorem.  
This implies that both $b'$ and $\nc'$ lie in the 
intersection region between the lower $\delta$-cone of $Q_0(a)$ and 
the right $\delta/2$-cone of $Q_4(c)$. 
Thus $a, b, c, \nc \in S$ satisfy the conditions of Lemma~\ref{lem:cd-close}, 
with the roles of $b$ and $c$ reversed:
$|b\nc| \le \sin(2\delta)/\sin(\pi/6+2\delta) \cdot |ac|$.


Arguments similar to the ones used 
in the proof of Lemma~\ref{lem:cd-close} 
show that $|c\nc| \le |ac|$. This along with $ |ac| \le |ab|$ (because $\overrightarrow{ac} \in \overrightarrow{Y_6}$) and 
the above inequality imply
\begin{eqnarray*}
p(a,b)  \le  |ac| + |c\nc| + p(\nc,b) 
 \le  2 |ab| + t |b\nc| 
 \le  t |ab|
\end{eqnarray*} 
for any $t$ satisfying the conditions stated by this lemma. 

%
\begin{figure}[ht]
    \centering
    \begin{tabular}{c@{\hspace{0.1\linewidth}}c}
        \includegraphics{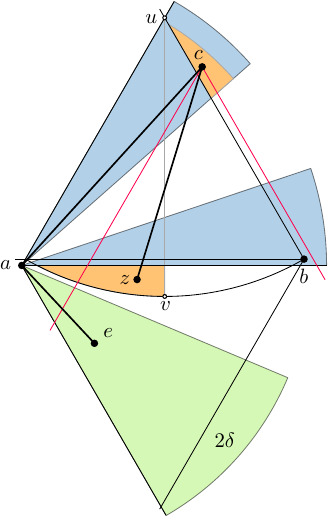} & 
        \includegraphics{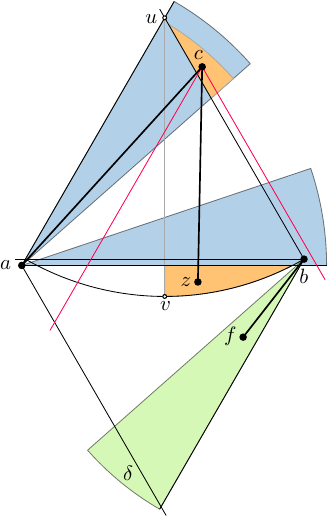} \\
        (a) & (b)  
    \end{tabular}
    \caption{The special situation, with $c \notin Q_2(b)$ and $\nc \in Q_5(a)$.
    (a) $\nc$ lies left of $uv$: $e \in Q_3(b)$.
    (b) $\nc$ lies right of $uv$: $f \in Q_5(a)$.}
    \label{fig:yao6case1}
\end{figure}
%
%
\vspace{-1em}
\paragraph{Case 3(b)} Assume now that $\nc \notin Q_0(a)$. Then $\nc \in  Q_5(a)$, as depicted in Figure~\ref{fig:yao6case1}. 
In this case $\nc$ lies in the disk sector $\canon{c}{b}$  (because $|c\nc| \le |cb|$) 
and below the horizontal line through $a$ (because $\canon{a}{c}$ is empty). 
This implies that there exists $\overrightarrow{ae} \in \overrightarrow{Y_6}$, with $e \in Q_5(a)$ and $|ae| \le |a\nc|$.
Similarly, there exists $\overrightarrow{bf} \in \overrightarrow{Y_6}$, with $f \in Q_3(b)$ and $|bf| \le |b\nc|$.
If $e$ lies above the lower $2\delta$-cone of $Q_5(a)$, then $\angle bae \le \pi/3-\delta$,
which leads to  \fbox{\emph{Induct}} and settles this case. 
Similarly, if $f$ lies above the lower $\delta$-cone of $Q_3(b)$, then $\angle abf \le \pi/3-\delta$,
which again leads to  \fbox{\emph{Induct}}. 
Otherwise, we show that the following lemma holds.

\begin{lemma}
Let $a, b, c, \nc \in S$ be in the configuration depicted in Figure~\ref{fig:yao6case1}, 
with $\overrightarrow{ac}, \overrightarrow{c\nc} \in \overrightarrow{Y_6}$. 
Let $\overrightarrow{ae}, \overrightarrow{bf} \in \overrightarrow{Y_6}$, with $e$ in the lower 
$2\delta$-cone of $Q_5(a)$ and $f$ in the lower $\delta$-cone of $Q_3(b)$. 
Then at least one of the following is true: (a) $e \in Q_3(b)$, or (b) $f \in Q_5(a)$.
\label{lem:special}
\end{lemma}
We defer a proof of Lemma~\ref{lem:special} to Section~\ref{sec:lemproof-special}.

Lemma~\ref{lem:special} guarantees that, if condition (a) holds, then $ae$ may not cross the 
lower ray bounding $Q_3(b)$. This case reduces 
to one of the cases depicted in Figure~\ref{fig:Pointd}, with $e$ playing the 
role of $c$ and the path passing under $ab$ rather than above.  Because $ae$ does not cross 
the lower ray bounding $Q_3(b)$, the special situation depicted in Figure~\ref{fig:yao6case0} 
(with $e$ playing the role of $c$) may not occur in this case. 
Similarly, condition (b) from Lemma~\ref{lem:special} reduces 
to one of the cases depicted in Figure~\ref{fig:Pointd}, with the roles of $a$ and 
$b$ reversed and with $f$ playing the role of $c$; the special situation depicted in Figure~\ref{fig:yao6case0} 
(with $bf$ playing the role of $ac$) may not occur in this case.  Having exhausted all cases, we conclude the proof. 
\end{proof}

Next we establish a lower bound on the spanning ratio of $Y_6$.

\begin{figure}[ht]
    \centering
    \begin{tabular}{c@{\hspace{0.1\linewidth}}c}
        \includegraphics{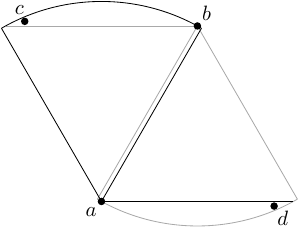} & 
        \includegraphics{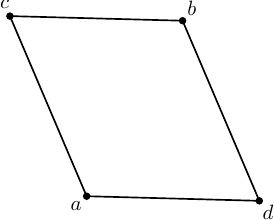} \\
        (a) & (b)  
    \end{tabular}
    \caption{$Y_6$ has spanning ratio at least $2$. (a) The lower bound construction. (b) The resulting $Y_6$ graph.}
    \label{fig:LowerBoundYao6}
  \end{figure}

\begin{theorem}
 The graph $Y_6$ has spanning ratio at least $2$.
\end{theorem}
\begin{proof}
We construct a lower bound example by extending the shortest path between two points $a$ and $b$.
Let $b \in Q_1(a)$ lie arbitrarily close to the cone boundary separating $Q_0(a)$ and $Q_1(a)$ (see Figure~\ref{fig:LowerBoundYao6}a).
Let $\varepsilon > 0$ be a small constant, and let $c \in Q_1(a)$ such that $|ac| = |ab| - \varepsilon$ and $|bc| \simeq |ab|$.
Similarly, let $d \in Q_4(b)$ such that $|bd| = |ab| - \varepsilon$ and $|ad| \simeq |ab|$.
Then the corresponding $Y_6$ graph is as depicted in Figure~\ref{fig:LowerBoundYao6}b.
Note that there are two shortest paths between $a$ and $b$, both of length $\simeq 2|ab|$.
\end{proof}

\section{Other lower bounds}
\label{sec:YaoLowerBounds}

In this section we provide lower bounds for all Yao graphs with at least six cones.
We do this by dividing this group of graphs in four families, depending on their number of cones.
We distinguish Yao graphs with $4x+2$ cones ($x \geq 1$), $4x+3$ cones, $4x+4$ cones, and $4x+5$ cones.
This division was introduced by Bose~\etal~\cite{bose2013spanning} to improve the analysis of $\Theta$-graphs.
It is based on the relation between the line perpendicular to the bisector of a cone and the cone boundaries.
For example, in a Yao graph with $4x+2$ cones, this line coincides with a cone boundary, whereas in a graph with $4x+4$ cones it coincides with the bisector of another cone.
For an overview of all bounds derived in this section, see Figure~\ref{fig:Graph} and Table~\ref{tab:results} in Section~\ref{sec:Conclusion}.

\subsection{Lower bound for \texorpdfstring{$\boldsymbol{Y_{4x+2}}$}{Y4x+2}}
In this section we provide a lower bound for Yao graphs with $4x+2$ cones ($x \geq 1$). 

\begin{theorem}
  For all $x \geq 1$, the graph $Y_{4x+2}$ has spanning ratio at least $1 + 2 \sin(\theta/2)$, where $\theta = 2\pi/(4x+2)$. 
\end{theorem}
\begin{proof}
  We construct the lower bound example by extending the shortest path between two vertices $a$ and $b$. We describe only how to extend one of the shortest paths between these vertices. To extend all shortest paths, the same modification is performed in each of the analogous cases (see Figure~\ref{fig:LowerBound4x+2}). 

  \begin{figure}[ht]
    \begin{center}
      \includegraphics{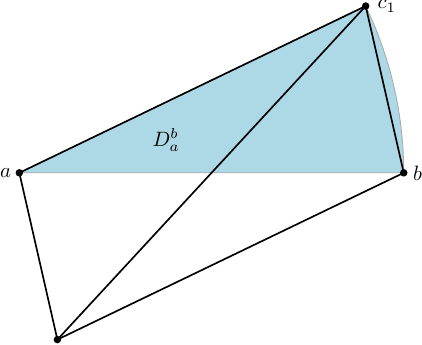}
    \end{center}
    \caption{The construction of the lower bound for $Y_{4x+2}$}
    \label{fig:LowerBound4x+2}
  \end{figure}

  First, we place $b$ arbitrarily close to a cone boundary of $a$. Next, we ensure that there is no edge between $a$ and $b$ by placing a vertex $c_1$ in the corner of \canon{a}{b} that is furthest from $b$ (see Figure~\ref{fig:LowerBound4x+2}). One of the shortest paths in the resulting graph visits $a$, $c_1$ and $b$. Thus, to obtain a lower bound for $Y_{4x+2}$, we compute the length of this path. 

  By construction, we have that $\angle c_1 a b = \theta$, hence we can express the various line segments as follows: 
  \begin{eqnarray*}
    |a c_1| &=& |a b| \\ \\ 
    |c_1 b| &=& 2 \sin \left( \frac{\theta}{2} \right) \cdot |a b|
  \end{eqnarray*}

  Hence, the total length of the shortest path is $|a c_1| + |c_1 b|$, which can be rewritten to \[\left( 1 + 2 \sin \left( \frac{\theta}{2} \right) \right) \cdot |a b|,\] proving the theorem. 
\end{proof}

\subsection{Lower bound for \texorpdfstring{$\boldsymbol{Y_{4x+3}}$}{Y4x+3}}
In this section we provide a lower bound for Yao graphs with $4x+3$ cones ($x \geq 1$). 

\begin{theorem}
  For all $x \geq 1$, the graph $Y_{4x+3}$ has spanning ratio at least \[1 + 2 \sin \left( \frac{3\theta}{8} \right) + 4 \frac{\left( \sin \left( \frac{13\theta}{16} \right) + \sin \left( \frac{19\theta}{16} \right)\right) \sin \left( \frac{\theta}{16} \right) \sin \left( \frac{3\theta}{8} \right)}{\sin(2\theta)},\] where $\theta = 2\pi/(4x+3)$. 
\end{theorem}
\begin{proof}
  We construct the lower bound example by extending the shortest path between two vertices $a$ and $b$ in three steps. We describe only how to extend one of the shortest paths between these vertices. To extend all shortest paths, the same modification is performed in each of the analogous cases (see Figure~\ref{fig:LowerBound4x+3}). 

  \begin{figure}[ht]
    \begin{center}
      \includegraphics{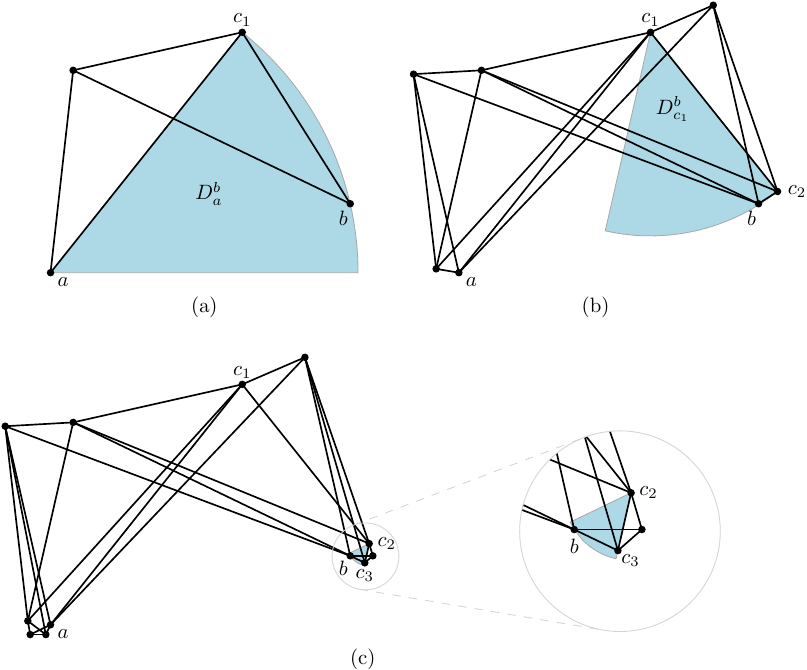}
    \end{center}
    \caption{The construction of the lower bound for $Y_{4x+3}$}
    \label{fig:LowerBound4x+3}
  \end{figure}

  First, we place $b$ such that the angle between $a b$ and the bisector of the cone of $a$ that contains $b$ is $\theta/4$. Next, we ensure that there is no edge between $a$ and $b$ by placing a vertex $c_1$ in the corner of \canon{a}{b} that is furthest from $b$ (see Figure~\ref{fig:LowerBound4x+3}a). Next, we place a vertex $c_2$ in the corner of \canon{c_1}{b} that is closest to $b$, since we need to ensure that there is no edge between $a$ and $c_2$ (see Figure~\ref{fig:LowerBound4x+3}b). Finally, we place a vertex $c_3$ in \canon{c_2}{b} such that \canon{c_3}{a} contains $b$ (see Figure~\ref{fig:LowerBound4x+3}c). This ensures that no shortcut to $a$ is created by $c_3$. One of the shortest paths in the resulting graph visits $a$, $c_1$, $c_2$, $c_3$ and $b$. Thus, to obtain a lower bound for $Y_{4x+3}$, we compute the length of this path. 

  \begin{figure}[ht]
    \begin{center}
      \includegraphics{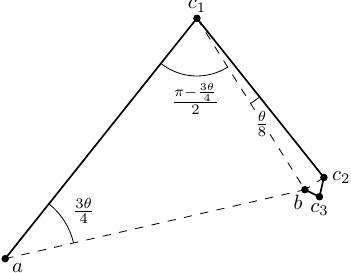}
    \end{center}
    \caption{The lower bound for $Y_{4x+3}$}
    \label{fig:LowerBoundComputation4x+3}
  \end{figure}

  By construction, we have that $\angle c_1 a b = 3\theta/4$. To compute the total length of the shortest path between $a$ and $b$, we also need $\angle c_2 c_1 b = \angle c_2 c_1 a - \angle b c_1 a$ (see Figure~\ref{fig:LowerBoundComputation4x+3}). Since $\angle c_2 c_1 a = \pi/2 - \theta/4$ and $\angle b c_1 a = (\pi - 3\theta/4) / 2$, it follows that $\angle c_2 c_1 b = \theta/8$. Finally, we need angles $\angle b c_3 c_2 = \pi - 2\theta$, $\angle c_3 c_2 b = \angle c_3 c_2 c_1 - \angle b c_2 c_1 = (\pi/2 + 3\theta/4) - (\pi/2 - \theta/16) = 13\theta/16$, and $\angle c_2 b c_3  = \pi - \angle b c_3 c_2 - \angle c_3 c_2 b = 19\theta/16$. Hence, we can express the various line segments as follows: 
  \begin{eqnarray*}
    |a c_1| &=& |a b| \\ \\ 
    |c_1 c_2| &=& |c_1 b| ~~=~~ 2 \sin \left( \frac{3\theta}{8} \right) \cdot |a b| \\ \\
    |c_2 b| &=& 2 \sin \left( \frac{\theta}{16} \right) \cdot |c_1 b| ~~=~~ 4 \sin \left( \frac{\theta}{16} \right) \sin \left( \frac{3\theta}{8} \right) \cdot |a b| \\ \\
    |c_2 c_3| &=& \frac{\sin \left( \frac{19\theta}{16} \right)}{\sin (\pi - 2\theta)} \cdot |c_2 b| ~~=~~ \frac{\sin \left( \frac{19\theta}{16} \right)}{\sin (2\theta)} \cdot|c_2 b| ~~=~~ 4\frac{\sin \left( \frac{19\theta}{16} \right) \sin \left( \frac{\theta}{16} \right) \sin \left( \frac{3\theta}{8} \right)}{\sin (2\theta)} \cdot |a b| \\ \\
    |c_3 b| &=& \frac{\sin \left( \frac{13\theta}{16} \right)}{\sin (\pi - 2\theta)} \cdot |c_2 b| ~~=~~ \frac{\sin \left( \frac{13\theta}{16} \right)}{\sin (2\theta)} \cdot|c_2 b| ~~=~~ 4\frac{\sin \left( \frac{13\theta}{16} \right) \sin \left( \frac{\theta}{16} \right) \sin \left( \frac{3\theta}{8} \right)}{\sin (2\theta)} \cdot |a b| 
  \end{eqnarray*}

  Hence, the total length of the shortest path is $|a c_1| + |c_1 c_2| + |c_2 c_3| + |c_3 b|$, which can be rewritten as
\[
 \left( 1 + 2 \sin \left( \frac{3\theta}{8} \right) + 4 \frac{\left( \sin \left( \frac{13\theta}{16} \right) + \sin \left( \frac{19\theta}{16} \right)\right) \sin \left( \frac{\theta}{16} \right) \sin \left( \frac{3\theta}{8} \right)}{\sin(2\theta)}\right) \cdot |a b|,
\]
proving the theorem. 
\end{proof}

\subsection{Lower bound for \texorpdfstring{$\boldsymbol{Y_{4x+4}}$}{Y4x+4}}
In this section we provide a lower bound for Yao graphs with $4x+4$ cones ($x \geq 1$). 

\begin{theorem}
  For all $x \geq 1$, the graph $Y_{4x+4}$ has spanning ratio at least
\[
  1 + 2 \sin \left( \frac{\theta}{2} \right) \left( 1 + \tan \left( \frac{\theta}{2} \right) \right),
\]
  where $\theta = 2\pi/(4x+4)$. 
\end{theorem}
\begin{proof}
  We construct the lower bound example by extending the shortest path between two vertices $a$ and $b$ in three steps. We describe only how to extend one of the shortest paths between these vertices. To extend all shortest paths, the same modification is performed in each of the analogous cases (see Figure~\ref{fig:LowerBound4x+4}). 

  \begin{figure}[ht]
    \begin{center}
      \includegraphics{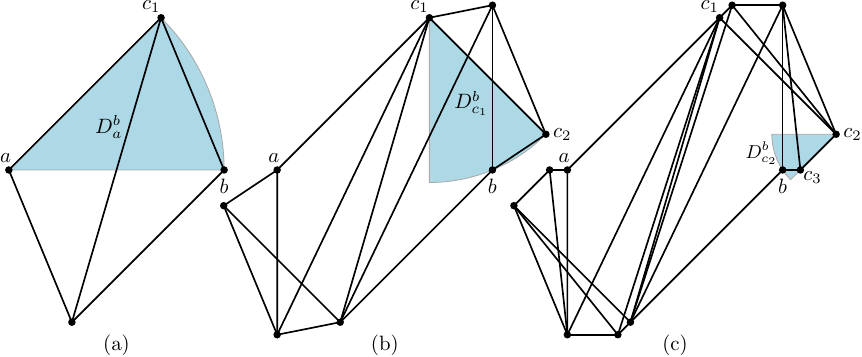}
    \end{center}
    \caption{The construction of the lower bound for $Y_{4x+4}$}
    \label{fig:LowerBound4x+4}
  \end{figure}

  First, we place $b$ arbitrarily close to a cone boundary of $a$. Next, we ensure that there is no edge between $a$ and $b$ by placing a vertex $c_1$ in the corner of \canon{a}{b} that is furthest from $b$ (see Figure~\ref{fig:LowerBound4x+4}a). Next, we place a vertex $c_2$ in the corner of \canon{c_1}{b} that is furthest from $a$ (see Figure~\ref{fig:LowerBound4x+4}b). Finally, we place a vertex $c_3$ on the intersection of \canon{c_2}{b} and \canon{b}{c_2} (see Figure~\ref{fig:LowerBound4x+4}c). This ensures that no shortcut to $a$ is created by $c_3$. One of the shortest paths in the resulting graph visits $a$, $c_1$, $c_2$, $c_3$ and $b$. Thus, to obtain a lower bound for $Y_{4x+4}$, we compute the length of this path. 

  \begin{figure}[ht]
    \begin{center}
      \includegraphics{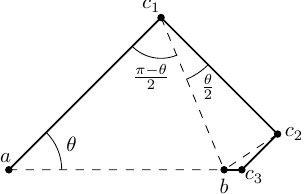}
    \end{center}
    \caption{The lower bound for $Y_{4x+4}$}
    \label{fig:LowerBoundComputation4x+4}
  \end{figure}

  By construction, we have that $\angle c_1 a b = \theta$. To compute the total length of the shortest path between $a$ and $b$, we also need $\angle c_2 c_1 b = \angle c_2 c_1 a - \angle b c_1 a$ (see Figure~\ref{fig:LowerBoundComputation4x+4}). Since $\angle c_2 c_1 a = \pi/2$ and $\angle b c_1 a = (\pi - \theta) / 2$, it follows that $\angle c_2 c_1 b = \theta/2$. Finally, we need angles $\angle b c_3 c_2 = \pi - \theta$, $\angle c_3 c_2 b = \angle c_3 c_2 c_1 - \angle b c_2 c_1 = \pi/2 - (\pi/2 - \theta/4) = \theta/4$, and $\angle c_2 b c_3 = \pi - \angle b c_3 c_2 - \angle c_3 c_2 b = 3\theta/4$. Hence, we can express the various line segments as follows: 
  \begin{eqnarray*}
    |a c_1| &=& |a b| \\ \\ 
    |c_1 c_2| &=& |c_1 b| ~~=~~ 2 \sin \left( \frac{\theta}{2} \right) \cdot |a b| \\ \\
    |c_2 b| &=& 2 \sin \left( \frac{\theta}{4} \right) \cdot |c_1 b| ~~=~~ 4 \sin \left( \frac{\theta}{4} \right) \sin \left( \frac{\theta}{2} \right) \cdot |a b| \\ \\
    |c_2 c_3| &=& \frac{\sin \left( \frac{3\theta}{4} \right)}{\sin (\pi - \theta)} \cdot |c_2 b| ~~=~~ \frac{\sin \left( \frac{3\theta}{4} \right)}{\sin (\theta)} \cdot|c_2 b| ~~=~~ 4\frac{\sin \left( \frac{3\theta}{4} \right) \sin \left( \frac{\theta}{4} \right) \sin \left( \frac{\theta}{2} \right)}{\sin (\theta)} \cdot |a b| \\ \\
    |c_3 b| &=& \frac{\sin \left( \frac{\theta}{4} \right)}{\sin (\pi - \theta)} \cdot |c_2 b| ~~=~~ \frac{\sin \left( \frac{\theta}{4} \right)}{\sin (\theta)} \cdot|c_2 b| ~~=~~ 4\frac{\sin \left( \frac{\theta}{4} \right) \sin \left( \frac{\theta}{4} \right) \sin \left( \frac{\theta}{2} \right)}{\sin (\theta)} \cdot |a b|
  \end{eqnarray*}

  Hence, the total length of the shortest path is $|a c_1| + |c_1 c_2| + |c_2 c_3| + |c_3 b|$, which can be rewritten to \[\left( 1 + 2 \sin \left( \frac{\theta}{2} \right) + 4\frac{\left( \sin \left( \frac{3\theta}{4} \right) + \sin \left( \frac{\theta}{4} \right) \right) \sin \left( \frac{\theta}{4} \right) \sin \left( \frac{\theta}{2} \right)}{\sin (\theta)} \right) \cdot |a b|,\] which is equal to 
\[ \left( 1 + 2 \sin \left( \frac{\theta}{2} \right) \left( 1 + \tan \left( \frac{\theta}{2} \right) \right) \right) \cdot |a b|, \]
 proving the theorem. 
\end{proof}

\subsection{Lower bound for \texorpdfstring{$\boldsymbol{Y_{4x+5}}$}{Y4x+5}}
In this section we provide a lower bound for Yao graphs with $4x+5$ cones ($x \geq 1$). 

\begin{theorem}
  For all $x \geq 1$, the graph $Y_{4x+5}$ has spanning ratio at least \[1 + 2 \sin \left( \frac{3\theta}{8} \right) + 4 \sin \left( \frac{5\theta}{16} \right) \sin \left( \frac{3\theta}{8} \right),\] where $\theta = 2\pi/(4x+5)$. 
\end{theorem}
\begin{proof}
  We construct the lower bound example by extending the shortest path between two vertices $a$ and $b$ in two steps. We describe only how to extend one of the shortest paths between these vertices. To extend all shortest paths, the same modification is performed in each of the analogous cases (see Figure~\ref{fig:LowerBound4x+5}). 

  \begin{figure}[ht]
    \begin{center}
      \includegraphics{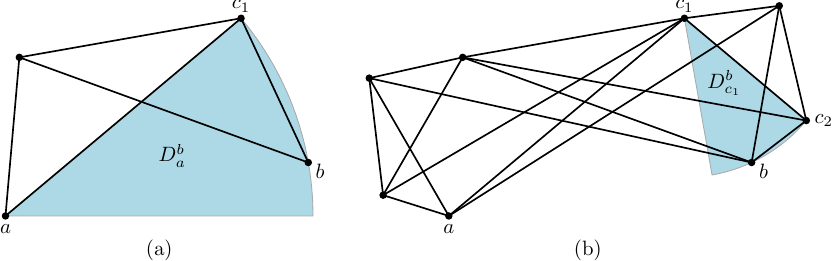}
    \end{center}
    \caption{The construction of the lower bound for $Y_{4x+5}$}
    \label{fig:LowerBound4x+5}
  \end{figure}

  First, we place $b$ such that the angle between $a b$ and the bisector of the cone of $a$ that contains $b$ is $\theta/4$. Next, we ensure that there is no edge between $a$ and $b$ by placing a vertex $c_1$ in the corner of \canon{a}{b} that is furthest from $b$ (see Figure~\ref{fig:LowerBound4x+5}a). Next, we place a vertex $c_2$ in the corner of \canon{c_1}{b} that is furthest from $b$ (see Figure~\ref{fig:LowerBound4x+5}b). One of the shortest paths in the resulting graph visits $a$, $c_1$, $c_2$ and $b$. Thus, to obtain a lower bound for $Y_{4x+5}$, we compute the length of this path. 

  \begin{figure}[ht]
    \begin{center}
      \includegraphics{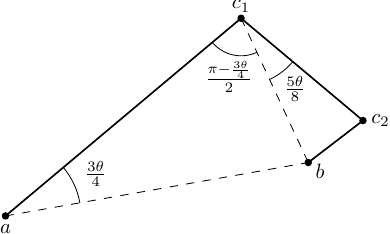}
    \end{center}
    \caption{The lower bound for $Y_{4x+5}$}
    \label{fig:LowerBoundComputation4x+5}
  \end{figure}

  By construction, we have that $\angle c_1 a b = 3\theta/4$. To compute the total length of the shortest path between $a$ and $b$, we also need $\angle c_2 c_1 b = \angle c_2 c_1 a - \angle b c_1 a$ (see Figure~\ref{fig:LowerBoundComputation4x+5}). Since $\angle c_2 c_1 a = \pi/2 + \theta/4$ and $\angle b c_1 a = (\pi - 3\theta/4) / 2$, it follows that $\angle c_2 c_1 b = 5\theta/8$. Hence, we can express the various line segments as follows: 
  \begin{eqnarray*}
    |a c_1| &=& |a b| \\ \\ 
    |c_1 c_2| &=& |c_1 b| ~~=~~ 2 \sin \left( \frac{3\theta}{8} \right) \cdot |a b| \\ \\
    |c_2 b| &=& 2 \sin \left( \frac{5\theta}{16} \right) \cdot |c_1 b| ~~=~~ 4 \sin \left( \frac{5\theta}{16} \right) \sin \left( \frac{3\theta}{8} \right) \cdot |a b|
  \end{eqnarray*}

  Hence, the total length of the shortest path is $|a c_1| + |c_1 c_2| + |c_2 b|$, which can be rewritten to \[\left( 1 + 2 \sin \left( \frac{3\theta}{8} \right) + 4 \sin \left( \frac{5\theta}{16} \right) \sin \left( \frac{3\theta}{8} \right)\right) \cdot |a b|,\] proving the theorem. 
\end{proof}

\subsection{The Yao-Yao graph with 5 cones is not a spanner}
\label{sec:YaoYao5}

One disadvantage of Yao graphs is that the maximum degree of a vertex might be $n - 1$.
For example, this happens when $n - 1$ points are spread evenly on a circle centered on the last point.
The \emph{Yao-Yao graph}, introduced by Li~\etal~\cite{li2002sparse}, solves this problem by first constructing the directed Yao graph, and then discarding all but the shortest incoming edge in each cone (ties are broken arbitrarily).
As a result, each vertex in the Yao-Yao graph with $k$ cones, denoted by $YY_k$, has maximum degree $2k$: one incoming and one outgoing edge per cone.
In the resulting graph, the directions of the edges are typically ignored.

Of course, discarding all these edges has a cost: the spanning ratio increases.
For a long time, it was unknown whether Yao-Yao graphs were even spanners.
The first answers to this question were negative.
Damian, El~Molla, and Pinciu~\cite{damian2009spanner,el2009yao} showed that, even though $Y_4$ and $Y_6$ are spanners, $YY_k$ is not a spanner for $k \leq 4$ and $k = 6$.
The first positive results followed soon afterwards, when Bauer and Damian~\cite{bauer2013infinite} proved that $YY_k$ is a spanner for all $k = 6x$, with $x \geq 6$.
The spanner status of all other Yao-Yao graphs is still open.

We close the gap among Yao-Yao graphs with six or fewer cones, by presenting a construction of a $YY_5$ graph whose stretch factor is unbounded.
Figure~\ref{fig:yy} shows the initial steps of constructing such a $YY_5$ graph, where the path between $a$ and $b$ can grow horizontally to the right by adding more points following the pattern, exceeding any bound on the stretch factor.

\begin{figure}[htb]
 \centering
 \begin{tabular}{c@{\hspace{0.1\linewidth}}c}
  \includegraphics{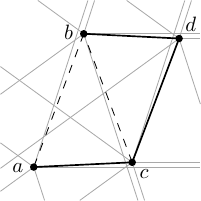} & 
  \includegraphics{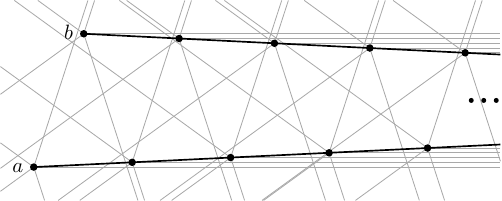} \\
  (a) & (b)  
 \end{tabular}
 \caption{The construction of a $YY_5$ graph with unbounded stretch factor. The gray lines are the boundaries of the cones. (a) The first four vertices. The dashed edges are in $Y_5$, but are discarded in favor of the bold edges. (b) The pattern continued to the right.}
 \label{fig:yy}
\end{figure}

The pattern begins with four vertices, positioned as in Figure~\ref{fig:yy}a. Vertex $b$ is placed close to the upper boundary of $Q_0(a)$, and $c$ is placed in $Q_0(a) \cap Q_4(b)$, near the intersection of their boundaries. Similarly, $d$ is placed in $Q_0(c) \cap Q_4(b)$ near the intersection of their boundaries. In the construction of the Yao-Yao graph on these points, the Yao edge $\overrightarrow{ba}$ is discarded since $\overrightarrow{ca}$ is shorter, and $\overrightarrow{cb}$ is discarded because $\overrightarrow{db}$ is shorter. By recursively applying this pattern with $c$ and $d$ in the roles of $a$ and $b$, we can eliminate the edge $cd$ and push the connecting edge arbitrarily far to the right. Since the distance between $a$ and $b$ remains the same, we can construct a $YY_5$ graph with arbitrarily large spanning ratio.

\begin{theorem}
 The graph $YY_5$ is not a constant spanner.
\end{theorem}

\section{Deferred proofs}

In this section, we prove the remaining lemmas required for the proof of Theorem~\ref{thm:Y6spanner}.

\subsection{Proof of Lemma~\ref{lem:cd-close}}
\label{sec:cd-close}

\textbf{Lemma~\ref{lem:cd-close}.}\emph{
Let $a, b, c, d \in S$ be as in Figure~\ref{fig:DeltaCones}a, with  
$\overrightarrow{bd} \in \overrightarrow{Y_6}$, 
$b, c \in Q_0(a)$ and $c, d \in Q_2(b)$. If both $c$ and $d$ lie above the lower rays bounding  
the upper 
$\delta$-cone of $Q_0(a)$ and the upper $2\delta$-cone of $Q_2(b)$, then for any $0 \le \delta \le \pi/9$,
\begin{equation*}
 |cd| \le \frac{\sin(2\delta)}{\sin(\pi/6+2\delta)}|ab|
\end{equation*}
}
\vspace{-1em}
\begin{proof}
Let $R$ be the intersection quadrilateral between the upper
$\delta$-cone of $Q_0(a)$ and the upper $2\delta$-cone of $Q_2(b)$. Let $u$ and $w$ be the top and bottom vertices of $R$, and $i$ and $z$ the left and right vertices of $R$, respectively. See Figure~\ref{fig:DeltaCones}a.

We first show that the diameter of $R$ is bounded above by $\max\{|ui|, |uw|\}$. Observe the following: (i) 
$\angle{uzw} = 2\pi/3-\delta \ge 5\pi/9$, therefore $|uw| > \max\{|uz|, |zw|\}$, (ii) $|iw| < |uz|$, and (iii) 
$|iz| \le \max\{|ui|, |uz|\}$, since $\angle{iuz} = \pi/3$ cannot exceed both other angles of $\triangle{uiz}$. It follows that 
the diameter of $R$ is no larger than $\max\{|ui|, |uw|\}$.


\begin{figure}[htbp]
\centering
\begin{tabular}{c@{\hspace{0.1\linewidth}}c}
\includegraphics{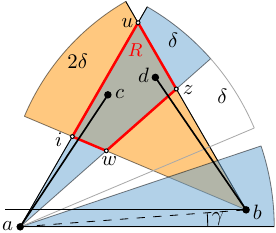} &
\includegraphics{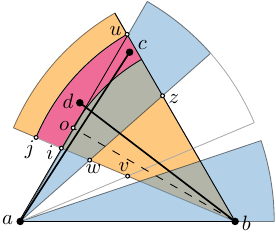} \\
(a) & (b) 
\end{tabular}
\caption{(a)~The configuration for Lemma~\ref{lem:cd-close}.
(b)~The distance $|cd|$ is still bounded by $|uv|$ when $d$ lies outside of $Q_0(a)$.
}
\label{fig:DeltaCones}
\end{figure}

Next we find an upper bound for $\max\{|uw|, |ui|\}$. 
Let $\gamma$ be the angle formed by $ab$ with the horizontal. 
We compute the quantities $|uw|$ and $|ui|$ as functions of $\delta$ and $\gamma$, and show that they are
maximized when $\gamma = 0$.
%
%
%
Set a coordinate system with the origin 
at $a$, and scale the point set $S$ so that $|ab| = 1$. 
Then the coordinates of $b$ are $(\cos \gamma, \sin \gamma)$. The point $u$ is at the 
intersection of the two lines passing through $a$ and $b$ 
with slopes $\tan \pi/3$ and $-\tan \pi/3$ respectively, given by
$y  =  \sqrt{3}x$ and $y  = -\sqrt{3}(x-\cos \gamma)+\sin \gamma$. 
Solving for $x$ and $y$ gives the coordinates of $u$
\begin{eqnarray*}
x_u  =  \frac{\sqrt{3}\cos \gamma+\sin \gamma}{2\sqrt{3}}, \quad\quad\
y_u  =  \frac{\sqrt{3}\cos \gamma+\sin \gamma}{2}.
\end{eqnarray*}

Similarly, the point $w$ is at the intersection of the line $L$ given by
$y = -\tan(\pi/3-2\delta)(x-\cos \gamma)+\sin \gamma$ and the line 
$y = \tan(\pi/3-\delta)x$; and $i$ is at the intersection of $L$ and the 
line $y = \tan(\pi/3)x = \sqrt{3}x$. 
Solving for $x$ and $y$ gives the coordinates of $w$ and $i$:

\begin{eqnarray*}
x_w  & = & \frac{\tan(\pi/3-2\delta)\cos \gamma+\sin \gamma}{\tan(\pi/3-\delta)+\tan(\pi/3-2\delta)}, \quad
y_w   =    \tan(\pi/3-\delta)x_w, \\
x_i  & = & \frac{\tan(\pi/3-2\delta)\cos \gamma+\sin \gamma}{\tan(\pi/3)+\tan(\pi/3-2\delta)}, \quad\quad\
y_i   =    \sqrt{3}x_i.
\end{eqnarray*}


\noindent
We can now compute $|uw| = \sqrt{(x_u-x_w)^2+(y_u-y_w)^2}$ as a function 
of $\gamma$ and $\delta$, and similarly for $|ui|$. 

\begin{figure}[htbp]
\centering
\begin{tabular}{c@{\hspace{0.08\linewidth}}c}
\includegraphics[width=0.4\linewidth]{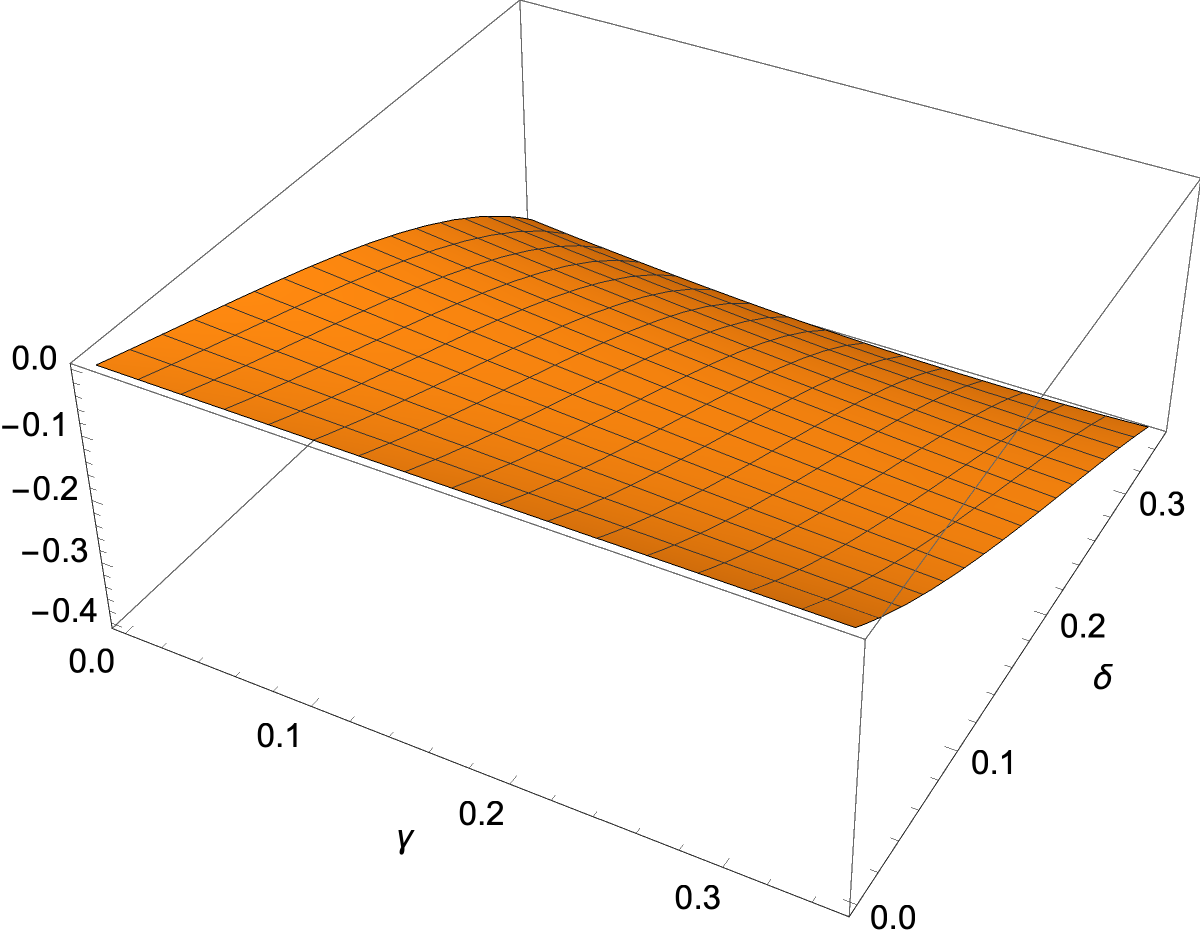} & 
\includegraphics[width=0.4\linewidth]{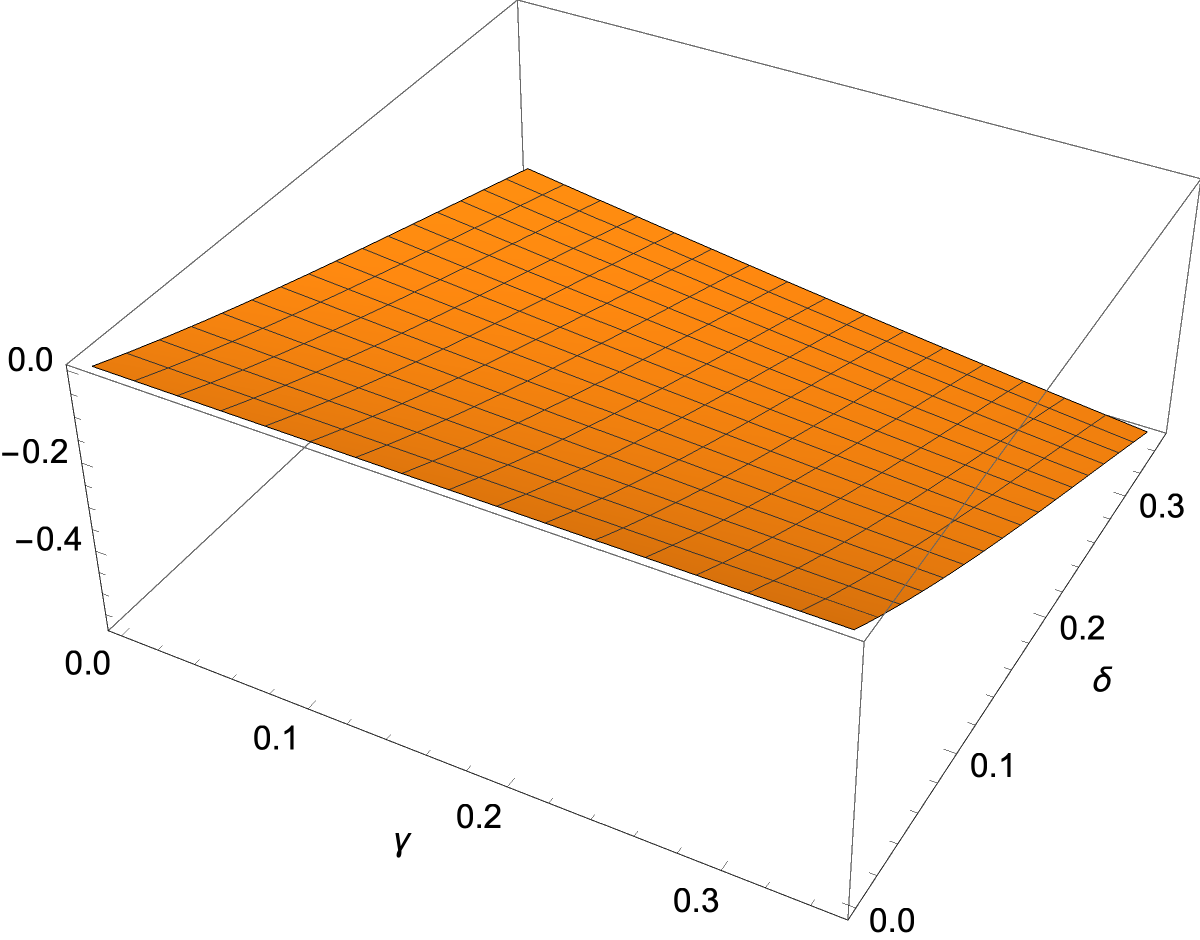} \\
(a) & (b) 
\end{tabular}
\caption{Derivatives of (a)  $|uw|$, and (b) $|ui|$ with respect to $\gamma$, for $\gamma, \delta \in [0, \pi/9]$.}
\label{fig:uv-plotd}
\end{figure}
We note that the derivatives of the functions $|uw|$ and $|ui|$ with respect to $\gamma$, 
depicted in~Figure~\ref{fig:uv-plotd}, are negative for $\gamma, \delta \in [0, \pi/9]$. 
This implies that $|uw|$ and $|ui|$ decrease as $\gamma$ increases, therefore their maximum value 
is achieved for $\gamma = 0$. 

We now set $\gamma = 0$ and find an upper bound on the quantities $|uw|$ and $|ui|$ (which, by 
the observation above, is an upper bound on $|uw|$ and $|ui|$ for any $\gamma \in [0, \pi/9]$). 
Let $v$ be the intersection point between the lower rays of the upper $2\delta$-cones of $Q_0(a)$ 
and $Q_2(b)$. Refer to~\ref{fig:DeltaCones}b. 
Observe that $\max\{|ui|, |uw|\} < |uv|$. To see this,
notice that $\angle{uiv} = 2\pi/3-2\delta \ge 4\pi/9$ and $\angle{iuv} = \pi/6$, therefore $\angle{uvi} \le 7\pi/18$. 
Thus $\angle{uiv}$ is the largest angle in $\triangle{uiv}$, which implies that $uv$ is the diameter of $\triangle{uiv}$, 
which in turn implies that $|uv| > \max\{|uw|, |ui|\}$. 
By the law of sines applied to $\triangle{uiv}$ (and the fact that $|au| = |ab| = 1$), we have 
$|uv| =  \sin(2\delta)/\sin(\pi/6+2\delta)$.


We have shown that the diameter of $R$ is bounded above by 
\[K_\delta = \frac{\sin(2\delta)}{\sin(\pi/6 + 2\delta)}.\] 

Thus, if both $c$ and $d$ are in $R$, 
then the claim of the lemma holds. 
Assume now that $d \notin R$, so $d$ lies above the upper ray bounding $Q_0(a)$. 
Let $o$ be the intersection point between the upper ray of $Q_0(a)$ and the bisector of $Q_2(b)$. 
Let $j$ be the intersection point between the lower ray bounding the upper $2\delta$-cone of $Q_2(b)$ and the circumference of $\canon{b}{u}$. If $\delta > \pi/12$ (and so $2\delta > \pi/6$), then $i$ lies below $o$, otherwise $i$ coincides with, or lies above $o$. We define the path $p(i)$ to be the line segment $io$ concatenated with the arc of the disk sector $\canon{b}{o}$, if $\delta > \pi/12$, or simply the arc of the disk sector $\canon{b}{i}$, if $\delta \le \pi/12$. Since $d$ is to the left of the line supporting $io$, we have that $|bo| < |bd| < |bc|$. (This latter inequality follows from the fact that $\overrightarrow{bd} \in \overrightarrow{Y_6}$.) This implies that $c$ also lies left of $p(i)$, so both $c$ and $d$ lie in the strip delimited by $\canon{b}{u}$, $p(i)$ and the two rays bounding the upper $2\delta$-cone of $Q_2(b)$. Thus $|cd|$ is no greater than the diameter of this strip, which we show to 
be no greater than $K_\delta$. 
For this, it suffices to show that $\max\{|ui|, |uj|, |ij|\} \le K_\delta$. 

As noted earlier, $|ui|$ decreases as $\gamma$ increases, therefore the maximum $|ui|$ value is achieved for 
$\gamma = 0$, and in this case we have shown that $|ui| < K_\delta$.
Similarly, it can be shown that $|uj|$ decreases as $\gamma$ increases, therefore the maximum $|uj|$ value is 
achieved for $\gamma = 0$. 
Next we set $\gamma = 0$ and show that $|uj| \le K_\delta$. 
From the isosceles triangle $\triangle buj$ we derive $\angle{ujv} = \pi/2-\delta$. 
Thus $\angle{juv} = \pi/2-\delta-\pi/6=\pi/3-\delta$ and $\angle{uvj} = \pi/6 + 2\delta \le \angle{ujv}$ for any 
$\delta \in [0, \pi/9]$. This along with the law of sines applied to $\triangle ujv$ implies that $|uj| \le |uv| = K_\delta$.


It remains to show that $|ij| \le K_\delta$. We will, in fact, show that $|ij| < |uj|$, which along with the conclusion above that 
 $|uj| \le K_\delta$, yields the desired result.  Angle $\angle{uij}$ is exterior to $\triangle uib$, therefore 
 $\pi/3 \le \angle{uij} \le \pi/3 + 2\delta$.  Earlier we showed that $\angle{ujv} = \pi/2-\delta \ge 7\pi/18$, 
 for any $\delta \le \pi/9$. It follows that $\angle{iuj} \le \pi-(7\pi/18+\pi/3) = 5\pi/18$ is the smallest 
 angle of $\triangle uij$, therefore $|ij| < |uj|$.  This completes the proof. 
\end{proof}

\subsection{Proof of Lemma~\ref{lem:special}}
\label{sec:lemproof-special}

%
\begin{figure}[ht]
    \centering
    \begin{tabular}{c@{\hspace{0.1\linewidth}}c}
        \includegraphics{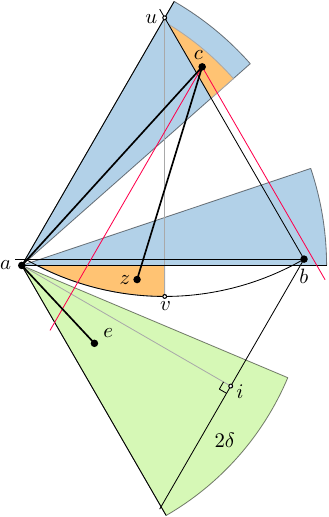} & 
        \includegraphics{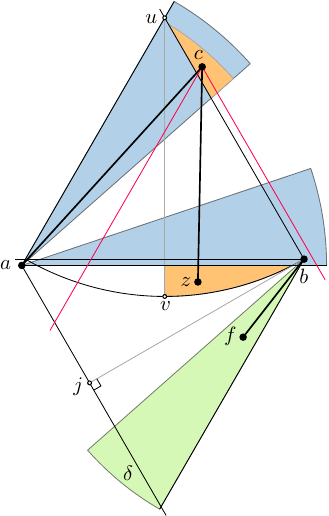} \\
        (a) & (b)  
    \end{tabular}
    \caption{The configuration for Lemma~\ref{lem:special}, with $c \notin Q_2(b)$ and $z \in Q_5(a)$.
    (a) $z$ lies left of $uv$: $e \in Q_3(b)$.
    (b) $z$ lies right of $uv$: $f \in Q_5(a)$.}
    \label{fig:yao6lemspecial}
\end{figure}
%

\textbf{Lemma~\ref{lem:special}.}\emph{
Let $a, b, c, z \in S$ be in the configuration depicted in Figure~\ref{fig:yao6lemspecial}, 
with $\overrightarrow{ac}, \overrightarrow{cz} \in \overrightarrow{Y_6}$. 
Let $\overrightarrow{ae}, \overrightarrow{bf} \in \overrightarrow{Y_6}$, with $e$ in the lower 
$2\delta$-cone of $Q_5(a)$ and $f$ in the lower $\delta$-cone of $Q_3(b)$. 
Then at least one of the following is true:  
\begin{enumerate}
\setlength{\itemsep}{0pt}
\item [(a)] $e \in Q_3(b)$
\item [(b)] $f \in Q_5(a)$ 
\end{enumerate}}
\vspace{-1em}
\noindent\begin{proof}
We define four intersection points $u$, $v$, $i$ and $j$ as follows: 
$u$ is at the intersection between the top rays of $Q_0(a)$ and $Q_2(b)$;
$v$ is at the intersection between the bisector of $\angle{aub}$ and 
the boundary of the disk sector $\canon{u}{b}$; 
$i$ is the foot of the perpendicular from $a$ on the lower ray of $Q_3(b)$; 
and $j$ is the foot of the perpendicular from $b$ on the lower ray of $Q_5(a)$.  
Refer to Figure~\ref{fig:yao6lemspecial}. 

Note that $|ae| \le |ai|$ implies condition (a), and $|bf| \le |bj|$ implies condition (b). 
We show that the first holds if $z$ lies to the left of or on $uv$, and the latter 
holds if $z$ lies to the right of or on $uv$ (and so at least one of the two conditions 
holds). We first show that $z \in \canon{u}{b}$. This follows immediately from the inequality 
$|uz| + |cb| < |cz| + |ub|$ (which can be derived using the triangle inequality 
twice on the triangles induced by the diagonals of $ucbz$), and the 
fact that $|cz| \le |cb|$ (because $\overrightarrow{cz} \in \overrightarrow{Y_6}$). It follows 
that $|uz| < |ub|$, therefore $z \in \canon{u}{b}$.

\paragraph{Condition (a)} Assume that $z$ lies to the left of $uv$ (as in Figure~\ref{fig:yao6lemspecial}a). 
Because $z \in \canon{u}{b}$ is below the horizontal through $a$, $\angle{azv}$ is obtuse and therefore $|az| \le |av|$ 
(equality holds when $z$ coincides with $v$). 
Also $|ae| \le |az|$, because $z$ and $e$ are in the same sector $Q_5(a)$ and 
$\overrightarrow{ae} \in \overrightarrow{Y_6}$.  It follows that $|ae| \le |av|$. We now show that 
$|av| \le |ai|$, which implies $|ae| \le |ai|$, thus settling this case. 

Let $\gamma \in [0, \delta]$ be the angle formed by $ab$ with the horizontal through $a$.
Then $\angle abi = \pi/3 - \gamma$ and $|ai| = |ab|\sin(\pi/3-\gamma)$.  
The law of sines applied to $\triangle uav$ tells us that
\[
    \frac{|av|}{\sin \pi/6} = \frac{|ua|}{\sin{\angle{uva}}} = \frac{|uv|}{\sin{\angle{uav}}}.
\]
Note that $|uv| = |ub| \le |ua|$, because $v$ lies on the circumference of $\canon{u}{b}$ and $a$ lies outside of this disk.
This along with the latter equality above yields $\angle{uav} \le \angle{uva}$.  
The sum of these two angles is $5\pi/6$ (recall that $uv$ is the 
bisector of $\angle{aub}$), therefore $\angle{uva} \ge 5\pi/12$. Also note that 
$\angle{uva} < \pi/2$, because $v$ lies strictly below the horizontal through $a$ 
(otherwise $d$ may not exist). It follows that $\sin{\angle{uva}} \ge \sin 5\pi/12$. Substituting this 
in the equality above yields $|av| \le |ua| \sin \pi/6 / \sin 5\pi/12$.
The law of sines applied to triangle $\triangle abu$ yields 
$|au| = |ab|\sin(\pi/3+\gamma)/\sin \pi/3$, which substituted in the previous equality yields 
\[
   |av| \le |ab| \frac{\sin(\pi/3+\gamma)\sin \pi/6}{\sin \pi/3 \sin 5\pi/12}.
\]
Thus the inequality $|av| \le |ai|$ holds for any $\gamma$ satisfying 
\[
   \frac{\sin(\pi/3+\gamma)\sin \pi/6}{\sin \pi/3 \sin 5\pi/12} \le \sin(\pi/3-\gamma).
\]
It can be easily verified that this inequality holds for any $\gamma \le \delta \le 23\pi/180$, and 
in particular for the $\delta$ values restricted by Lemma~\ref{lem:main}.

\paragraph{Condition (b)} Assume now that $z$ lies to the right of $uv$ (as in Figure~\ref{fig:yao6lemspecial}b). 
In this case $|bf| \le |bz| \le |bv|$. We now show that $|bv| \le |bj|$, which implies 
$|bf| \le |bj|$, thus settling this case.
From the right triangle $\triangle baj$ with angle $\angle baj = \pi/3+\gamma$, we derive 
$|bj| = |ab|\sin(\pi/3+\gamma)$. Next we derive an upper bound on $|bv|$. 
From the isosceles triangle $\triangle vub$, having angle $\angle vub = \pi/6$, we derive 
$|bv| = 2|bu| \sin \pi/12$. 
The law of sines applied to triangle $\triangle uab$ gives us  
$|ub| = |ab|\sin(\pi/3-\gamma)/\sin \pi/3$, 
which substituted in the previous equality yields  
$|bv| = 2|ab|\sin(\pi/3-\gamma)\sin \pi/12/\sin \pi/3$. 
Thus the inequality $|bv| \le |bj|$ holds for any $\gamma$ value satisfying
\[
\frac{2\sin(\pi/3-\gamma)\sin \pi/12}{\sin \pi/3} \le  \sin(\pi/3+\gamma).
\] 
It can be verified that this inequality holds for any $\gamma \le \delta \le \pi/3$, 
and in particular for the $\delta$ values restricted by Lemma~\ref{lem:main}.
\end{proof}

\section{Conclusion}
\label{sec:Conclusion}

\begin{figure}[htbp]
 \centering
 \includegraphics{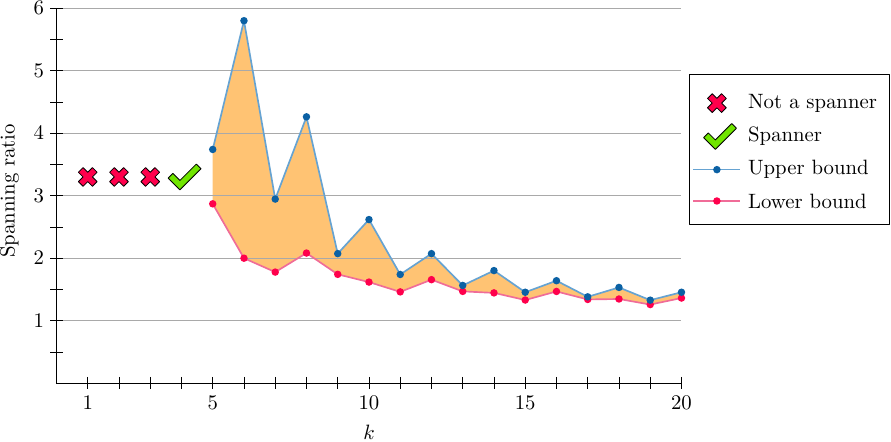}
 \caption{A plot showing the best known bounds on the spanning ratio of Yao graphs.}
 \label{fig:Graph}
\end{figure}

The main contributions of this paper are threefold. First we prove an upper bound of $3.74$ on the spanning ratio of $Y_5$. This answers the question of which Yao graphs are spanners, by establishing that $Y_k$ is a spanner if and only if $k \geq 4$. In addition, the new upper bound falls just below the lower bound of $3.79$ established for $\Theta_5$~\cite{bose2013theta5}, marking the first separation between the spanning ratio of a Yao graph and its peer $\Theta$-graph. For all other $k$ it is still unknown which of $Y_k$ or $\Theta_k$ has a better spanning ratio. By exploiting the asymmetry of Yao graphs with an odd number of cones, we also improve the upper bound on the spanning ratio of all graphs $Y_k$, for odd $k \ge 7$. The improvement is particularly significant for low values of $k$ -- for instance, we reduce the the upper bound on the spanning ratio of $Y_7$ from 7.562 to $2.946$. Figure~\ref{fig:Graph} and Table~\ref{tab:results} summarize the currently known bounds on the spanning ratio of Yao graphs, including the results of this paper.

Second, we reduce the upper bound on the spanning ratio of $Y_6$ from $17.64$~\cite{damian2012yao} to $5.8$. We attribute this significant improvement to our direct approach of proving that $Y_6$ is a spanner, instead of proving that $Y_6$ spans the edges of another spanner. We also present the first lower bounds on the spanning ratio of $Y_k$, for $k \ge 6$. The gaps left between the lower and upper bounds are fairly small, especially for odd $k$. This is mostly due to our improved upper bound for odd Yao graphs, leading us to believe that the upper bounds on the spanning ratios of $Y_k$, for even $k$, could be further reduced. In fact, we conjecture that the spanning ratio of $Y_6$ matches the spanning ratio of 2 established for $\Theta_6$.

Finally, we bring some light into the obscure world of Yao-Yao graphs, by proving that $YY_5$ is not a spanner. In this area, many interesting questions remain unanswered. For instance, even though Bauer and Damian~\cite{bauer2013infinite} showed that there exist an infinite number of Yao-Yao graphs that are spanners, the first member of this class is $YY_{36}$, and nothing is known about the spanning properties of $YY_7$ through $YY_{35}$. In particular, it would be interesting to know whether, as with Yao and $\Theta$-graphs, there is a constant $k$ such that all Yao-Yao graphs with more than $k$ cones are spanners, or whether there exist an infinite number of Yao-Yao graphs that are not spanners.

\newcolumntype{M}[1]{>{\centering\arraybackslash$\displaystyle}m{#1}<{$}}
\newcolumntype{N}{@{}m{0pt}@{}}

\begin{table}[htbp]
\caption {Lower and upper bounds on the spanning ratio of Yao graphs.} \label{tab:results}
\begin{center}\renewcommand{\arraystretch}{1.6}
\begin{tabular}{|M{4.25em}|M{24em}|M{7.75em}|N}
  \hline
  k & \text{Lower bound} & \text{Upper bound} \\\hline\hline
  k=2,3 & \infty~\cite{el2009yao} & \text{Not a spanner} \\\hline
  k=4 & \text{Open} & 16(23+13\sqrt{2}) \approx 663~\cite{bose2012pi} \\\hline
  k=5 & 2.87 & 2 + \sqrt{3} \approx 3.74 \\\hline
  k=6 & 2 & 5.8 \\\hline
  k=4x+2 & 1 + 2 \sin\left(\frac{\theta}{2}\right) & \frac{1}{1-2\sin\left(\frac{\theta}{2}\right)}~\cite{bose2012piArxiv} & \\[1em]\hline
  k=4x+3 & 1 + 2 \sin \left( \frac{3\theta}{8} \right) + 4 \frac{\left( \sin \left( \frac{13\theta}{16} \right) + \sin \left( \frac{19\theta}{16} \right)\right) \sin \left( \frac{\theta}{16} \right) \sin \left( \frac{3\theta}{8} \right)}{\sin(2\theta)} & \frac{1}{1-2\sin\left(\frac{3\theta}{8}\right)} & \\[1em]\hline
  k=4x+4 & 1 + 2 \sin \left( \frac{\theta}{2} \right) \left( 1 + \tan \left( \frac{\theta}{2} \right) \right) & \frac{1}{1-2\sin\left(\frac{\theta}{2}\right)}~\cite{bose2012piArxiv} & \\[1em]\hline
  k=4x+5 & 1 + 2 \sin \left( \frac{3\theta}{8} \right) + 4 \sin \left( \frac{5\theta}{16} \right) \sin \left( \frac{3\theta}{8} \right) & \frac{1}{1-2\sin\left(\frac{3\theta}{8}\right)} & \\[1em]\hline
\end{tabular}
\end{center}
\end{table}

\bibliographystyle{plain}
\bibliography{contyao}

\newpage
\appendix
\section{Lower bound coordinates}
\label{app:LB-coordinates}

The following table lists the coordinates of the points in the $Y_5$ graph shown in Figure~\ref{fig:frac-4b}, whose spanning ratio is more than 2.87.

\begin{table}[h]
\centering
\begin{tabular}{@{(}r@{, }r@{)\hspace{1em} }l@{\hspace{8em}(}r@{, }r@{)}}
0 & 0 & $a$ & 341 & 264 \\
252 & 82 & $b$ & -179 & 97 \\
130 & 230 & $c'$ & -180 & 112 \\
12 & 193 & $d'$ & -91 & -75 \\
30 & 302 & & 316 & 36 \\
293 & 269 & & 352 & 229 \\
321 & 229 & & 303 & 297 \\
-143 & 130 & & -167 & 63 \\
-143 & 80 & & -167 & 147 \\
193 & 384 & & -26 & -75 \\
158 & 367 & & 371 & 213 \\
-135 & 272 & & 51 & 310 \\
-91 & 287 & & -176 & 37 \\
-153 & -55 & & 344 & 274 \\
371 & 75 & & -189 & 105 \\
410 & 115 & & 99 & 320 \\
334 & 276 & & -15 & 284
\end{tabular}
\caption{Coordinates of the points in Figure~\ref{fig:frac-4b}}
\end{table}


\medskip
\noindent
{\bf Acknowledgement.} We thank Davood Bakhshesh for pointing out a flaw in the arguments of Lemma 13 in~\cite{yao56-jocg-15}, which we have corrected in this document. 

\end{document}